\documentclass[a4paper, twocolumn]{article}
\usepackage[utf8]{inputenc}
\usepackage[T1]{fontenc}
\usepackage[english]{babel}
\usepackage{hyperref}
\usepackage[table]{xcolor}
\usepackage{amsmath}
\usepackage{amsfonts}
\usepackage{amsthm}
\usepackage{tikz}
\usepackage{subcaption}
\usepackage[shortcuts]{extdash}
\usepackage{diagbox}
\usepackage[inline]{enumitem}
\usepackage{multirow}
\usepackage{float}
\usepackage[section]{placeins}
\usepackage{booktabs}
\usepackage[misc]{ifsym}
\usepackage{authblk}
\usepackage{biblatex}

\usetikzlibrary{decorations.pathmorphing,calc,math,3d,external}
\theoremstyle{plain}
\newtheorem{lemma}{Lemma}
\theoremstyle{definition}
\newtheorem{definition}{Definition}

\title{Pebbles, graphs and equilibria: higher order shape descriptors for sedimentary particles}

\author[1,2]{Balázs Ludmány\thanks{\href{mailto:ludmany.balazs@cloud.bme.hu}{ludmany.balazs@cloud.bme.hu}}}
\author[1,3]{Gábor Domokos}

\affil[1]{ELKH\==BME Morphodynamics Research Group}
\affil[2]{Dept. of Control Engineering and Information Technology, Budapest University of Technology and Economics, 2. Magyar tudosok Blvd., H\==1117, Budapest}
\affil[3]{Dept. of Morphology and Geometric Modeling, Budapest University of Technology and Economics 1\==3. Műegyetem rakpart, H\==1111, Budapest}

\begin{document}

\onecolumn
\maketitle

\begin{abstract}
  While three\=/dimensional measurement technology is spreading fast, its meaningful application to sedimentary geology still lacks content. Classical shape descriptors (such as axis ratios, circularity of projection) were not inherently three\=/dimensional, because no such technology existed. Recently a new class of three\=/dimensional descriptors, collectively referred to as \emph{mechanical descriptors} has been introduced and applied for a broad range of sedimentary particles. \emph{First order} mechanical descriptors (registered for each pebble as a pair $\{S,U\}$ of integers), refer to the respective \emph{numbers} of stable and unstable static equilibria and can be reliably detected by hand experiments. However, they have limited ability of distinction as the majority of coastal pebbles fall into primary class $\{S,U\}=\{2,2\}$. \emph{Higher order} mechanical descriptors offer a more refined distinction.  However for the extraction of these descriptors (registered as \emph{graphs} for each pebble) hand measurements are not an option and even computer\=/based extraction from 3D scans offers a formidable challenge. Here we not only describe and implement an algorithm to perform this task, but also apply it to a collection of 271 pebbles with various lithologies, illustrating that the application of higher order descriptors is a viable option for geomorphologists. We also show that the so-far uncharted connection between the two known secondary descriptors, the so-called    \emph{Morse\==Smale graph} and the \emph{Reeb\=/graph} can be established via a \emph{third order descriptor} which we call the master graph.
  

  \paragraph{Acknowledgments} Support of the NKFIH Hungarian Research Fund, grants 134199 and of the NKFIH Fund TKP2021 BME\=/NVA, carried out at the Budapest University of Technology and Economics, is kindly acknowledged.
\end{abstract}

\twocolumn

\setcounter{tocdepth}{2}
\tableofcontents

\section{Introduction: Motivation and main goals}\label{sec:intro}

The shape of  sedimentary particles carries an infinite amount of information and we know that some portion of this is highly relevant for detecting the provenance of the particle~\cite{domokos_natural, szabo_universal,szabo_mars}. However, picking this relevant portion  may not always be trivial:  this is the observer's prerogative, manifested in the choice of \emph{shape descriptors}. This choice is a trade\=/off between the observer's ability to reliably measure the chosen shape descriptor and the observer's desire to extract maximal information.  \emph{Classical descriptors}, such as axis ratios and the isoperimetric ratio (which we discuss in subsection~\ref{ss:classical}) rely on hand measurement and two\=/dimensional image analysis. They were picked by geomorphologists in the absence of computerized three\=/dimensional tools, based on the mentioned trade\=/off. One, recently introduced set of shape descriptors, called \emph{first order mechanical descriptors}~\cite{szabo_pebbles, varkonyi_gomboc}, registered for each particle as a pair of integers, can already be regarded as three\=/dimensional descriptors but still relies on hand measurements. First order mechanical descriptors proved to be particularly informative~\cite{domokos_natural, domokos_rocking,  szabo_universal} in describing the provenance of particles (see subsection~\ref{sss:higher}, and see also Figure~\ref{fig:primary}.) Their generalizations, also known as \emph{second order mechanical descriptors}~\cite{holmes, domokos_morse_smale, ludmany2021morsesmale, nicolaescu_counting}  registered for each particle as a graph (either a so-called Morse\==Smale graph or a so-called  Reeb\=/graph, see subsection~\ref{sss:higher}, see also Figure~\ref{fig:ms_and_reeb}), also appear to be promising tools in sedimentary geomorphology. However, they are fully three\=/dimensional descriptors and their measurement has been prohibitive until now and they could not be applied in either laboratory or field studies.

The challenge of measuring second order mechanical shape descriptors is three\=/fold: 
\begin{enumerate}
    \item \emph{Mathematical challenge:} There exist at least two different kinds of secondary descriptors: Morse\==Smale graphs~\cite{holmes, domokos_morse_smale, ludmany2021morsesmale} and Reeb\=/graphs~\cite{nicolaescu_counting}. The connection between these two secondary descriptors has not been determined until now: that is, it was not clear how the two classification schemes are related, whether one of them may be regarded as a refinement of the other or not.     

    \item \emph{Algorithmic challenge:} The extraction of second order descriptors from 3D datasets is far from trivial. While related problems have been solved in image processing on surfaces defined in orthogonal coordinates \cite{edelsbrunner_morse_smale}, the spherical version of this method (to be used on particle surfaces) has not been developed.
    
    \item \emph{Technological challenge:} Obtaining 3D (scanned) datasets for sedimentary particles in a fast and reliable manner appeared to be, until very recently, quite challenging. 3D scanning technologies did not offer the option to measure an object quickly and reliably on the full spherical horizon.
\end{enumerate}

The last mentioned (technological) challenge appears to be resolved: as 3D measurement technology is spreading fast~\cite{Rodriguez_3D_2013, LATHAM_3D_2008, Sun_3D_2014}, it may become soon the de~facto standard also in sedimentary geomorphology~\cite{steer2022, havasi}. Encouraged by these developments, our paper takes aim at the first two challenges. In particular we offer the following:
\begin{enumerate}
    \item \emph{Mathematical results:} In subsections~\ref{sec:master} and \ref{sec:poly_master} we  introduce a \emph{third order mechanical descriptor} which we call the \emph{master graph} which establishes a meaningful connection between second order mechanical descriptors, the Reeb\=/graph and the Morse\==Smale graph. We show that neither of them can be derived from the other, but both can be derived from the master graph. We will also give a complete third\=/order description of the geologically most relevant primary class $\{S,U\}=\{2,2\}$, containing the majority of all coastal pebbles (see also Table~\ref{tbl:distribution} and Figure~\ref{fig:subdivision}). In particular, in subsection~\ref{sss:higher} we will formulate Lemma~\ref{lem:cardinality}, claiming that there exist 3 tertiary classes in the primary class $\{S,U\}=\{2,2\}$ and we prove this claim in Section~\ref{sec:catalog}.
    
    \item \emph{Algorithmic results and application}: In Section~\ref{sec:polyhedral} we describe a reliable tool to determine higher order mechanical descriptors based on 3D point clouds obtained from scans. Here we will consider all three relevant types of graphs: Morse\==Smale graphs, Reeb\=/graphs and the master graph.  We illustrate our algorithm in the Supplementary material where we show both first, second and third\=/order mechanical descriptors for 271 scanned pebbles of various lithologies. On these 271 pebbles we identified 29 primary equilibrium classes, 69 distinct Reeb\=/graphs, 62 distinct Morse\==Smale graphs and 115 distinct master graphs. In particular, despite the fact that over 50 pebbles belong to the primary class $\{2,2\}$, inside it we only identified one Reeb\=/graph, one Morse\==Smale graph and one single master graph.
\end{enumerate}

Despite the formidable difficulties of their measurement, higher order mechanical descriptors appear to be an interesting and inviting tool for geomorphology: they are naturally encoded in the pebble shape, that is, when we use them in the description of the pebble, we do not add any arbitrary, man\=/made information. They carry deep, essential information on the shape and its evolution and thus they might help to uncover new, surprising connections between pebbles and pebble populations. It is not a coincidence that these concepts have been applied in image processing and morphology. In this paper we offer the above mentioned results and algorithmic tools as the first step towards the geological application of these deep, natural shape descriptors.

The structure of the paper is the following: In Section~\ref{sec:concepts} we give intuitive definitions of the basic concepts and state the above mentioned Lemma~\ref{lem:cardinality}. In Section~\ref{sec:classification} we discuss the previous concepts more rigorously on smooth, convex shapes. In Section~\ref{sec:polyhedral} we interpret the same concepts in the context of convex polyhedra as models of scanned particles. In Section~\ref{sec:catalog} we prove Lemma~\ref{lem:cardinality} and other mathematical results. Section~\ref{sec:flocks} describes our algorithmic results, in particular, we show how we can extract higher order classes from natural shapes based on 3D scans.

\section{Basic concepts}\label{sec:concepts}
\subsection{Shape catalogs}\label{ss:catalogs}

If the value of  a shape descriptor may assume an interval of real numbers (such as the value of roundness) then we refer to it as a \emph{continuous descriptor}, whereas if it defines discrete classes (i.e.~its set of values is discrete) then we call it a \emph{shape catalog} (Zingg classes, discussed below, are an example). Catalogs (not necessarily of shapes) are fundamental tools of scientific progress and they have been applied broadly in physics (conductors/semiconductors/insulators), chemistry (periodic table) and biology (taxonomy). The main advantage of a shape catalog, compared to a continuous shape descriptor, is that recorded datasets are exact, unambiguous and they can be easily interpreted, compared and understood. Catalogs may be either \emph{natural} or \emph{artificial}, depending on whether classes are separated based on some natural, or some man\=/made condition. Catalogs may be, depending on the number of classes, either \emph{finite} or \emph{infinite}. Catalogs may be either \emph{complete} or \emph{incomplete}, depending on whether each class contains shapes or not. Catalogs may be either \emph{biased} or \emph{uniform}, depending on the statistical distribution of natural shapes among classes (in biased catalogs this distribution is non\=/uniform). In case of shape catalogs it is easy to see that from the geophysical point of view we seek natural, biased classifications (whether or not they are complete may be more of mathematical interest).

While we did not apply higher order mechanical descriptors in field studies yet, still, by relying on a small laboratory dataset of 52 pebbles we argue that their application carries substantial potential for geophysical insight as they offer \emph{biased natural catalogs} of sedimentary shapes, meaning they offer naturally defined, discrete classification schemes where the majority of natural shapes is contained in very few classes. Such strong bias can not only motivate the search for particular (rare) natural shapes, it can also offer clues about the underlying evolution process. The first such strongly biased natural catalog is associated with first\=/order mechanical shape descriptors and led to the discovery of a fundamental monotonic trend in natural shape evolution, the monotonic decrease of static balance points~\cite{domokos_natural, szabo_universal}. Below we briefly review classical shape descriptors as well as mechanical shape descriptors and explain the concept of biased catalogs in more detail.

\subsection{Classical descriptors}\label{ss:classical}

The most established shape descriptors are, without doubt, \emph{axis ratios}~\cite{zingg} which, for an ellipsoid with axes $a>b>c$ may be written as $y^1=c/b, y^2=b/a$ and for non\=/ellipsoidal shapes an approximating ellipsoid is considered. An alternative geophysical shape descriptor is \emph{roundness}~\cite{krumbein} which is commonly measured as the isoperimetric ratio $0 \leq I \leq 1$ of the pebble's contour~\cite{szabo_mars, szabo_universal}. We will  refer to the axis ratios~$y^1,y^2$ and the isoperimetric ratio~$I$ as \emph{classical descriptors}. Classical descriptors are, by definition, real numbers defined on a continuous domain, so, in the sense defined in subsection~\ref{ss:catalogs}, they are \emph{continuous descriptors}. Such descriptors admit the comparison of sample averages, however, they do not immediately provide a classification or catalog for natural shapes which could be based only on integer\=/type descriptors. While we are not aware of any existing catalog for the isoperimetric ratio~$I$, in case of axis ratios it is apparent that there is a need for such catalogs: by introducing the thresholds at $y^1=y^2=2/3$, Zingg~\cite{zingg} created the first such system, subdividing all shapes into the four classes called \emph{discs, spheres, blades} and \emph{rods}.

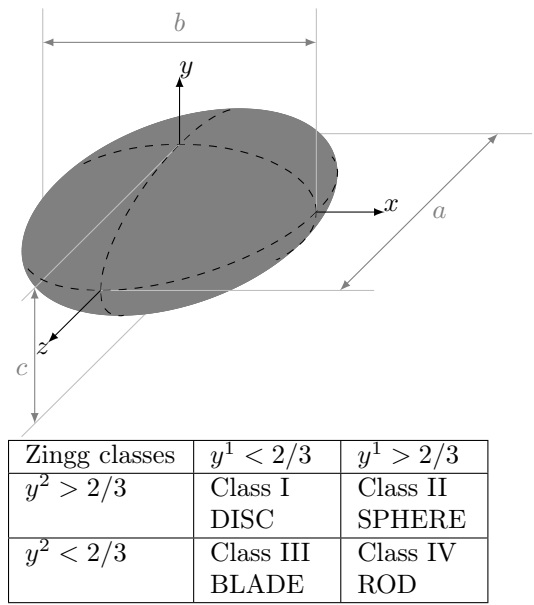
\begin{figure}[!ht]
  \centering
  \begin{tikzpicture}[scale=.9]
    \draw[gray!50!white] (0,0,-3) -- (4,0,-3);
    \draw[gray!50!white] (-2,0,0) -- (-2,3,0);
    \draw[gray!50!white] (0,-1,0) -- (0,-1,6);
    \fill[gray] (0, 0, -3.0) ellipse (0.0 and 0.0);
\fill[gray] (0, 0, -2.9) ellipse (0.5120763831912405 and 0.25603819159562025);
\fill[gray] (0, 0, -2.8) ellipse (0.7180219742846008 and 0.3590109871423004);
\fill[gray] (0, 0, -2.7) ellipse (0.8717797887081344 and 0.4358898943540672);
\fill[gray] (0, 0, -2.6) ellipse (0.9977753031397176 and 0.4988876515698588);
\fill[gray] (0, 0, -2.5) ellipse (1.1055415967851334 and 0.5527707983925667);
\fill[gray] (0, 0, -2.4) ellipse (1.2000000000000002 and 0.6000000000000001);
\fill[gray] (0, 0, -2.3) ellipse (1.284090685617215 and 0.6420453428086075);
\fill[gray] (0, 0, -2.2) ellipse (1.359738536958076 and 0.679869268479038);
\fill[gray] (0, 0, -2.1) ellipse (1.42828568570857 and 0.714142842854285);
\fill[gray] (0, 0, -2.0) ellipse (1.4907119849998598 and 0.7453559924999299);
\fill[gray] (0, 0, -1.9) ellipse (1.5477582354991866 and 0.7738791177495933);
\fill[gray] (0, 0, -1.7999999999999998) ellipse (1.6 and 0.8);
\fill[gray] (0, 0, -1.7) ellipse (1.6478942792411033 and 0.8239471396205517);
\fill[gray] (0, 0, -1.5999999999999999) ellipse (1.6918103387266028 and 0.8459051693633014);
\fill[gray] (0, 0, -1.5) ellipse (1.7320508075688772 and 0.8660254037844386);
\fill[gray] (0, 0, -1.4) ellipse (1.7688665548562132 and 0.8844332774281066);
\fill[gray] (0, 0, -1.2999999999999998) ellipse (1.80246744461277 and 0.901233722306385);
\fill[gray] (0, 0, -1.2) ellipse (1.8330302779823362 and 0.9165151389911681);
\fill[gray] (0, 0, -1.0999999999999999) ellipse (1.8607047649270483 and 0.9303523824635241);
\fill[gray] (0, 0, -1.0) ellipse (1.8856180831641267 and 0.9428090415820634);
\fill[gray] (0, 0, -0.8999999999999999) ellipse (1.9078784028338913 and 0.9539392014169457);
\fill[gray] (0, 0, -0.7999999999999998) ellipse (1.9275776393067947 and 0.9637888196533974);
\fill[gray] (0, 0, -0.6999999999999997) ellipse (1.9447936194419762 and 0.9723968097209881);
\fill[gray] (0, 0, -0.5999999999999996) ellipse (1.9595917942265426 and 0.9797958971132713);
\fill[gray] (0, 0, -0.5) ellipse (1.9720265943665387 and 0.9860132971832694);
\fill[gray] (0, 0, -0.3999999999999999) ellipse (1.9821424996424675 and 0.9910712498212337);
\fill[gray] (0, 0, -0.2999999999999998) ellipse (1.98997487421324 and 0.99498743710662);
\fill[gray] (0, 0, -0.19999999999999973) ellipse (1.9955506062794355 and 0.9977753031397177);
\fill[gray] (0, 0, -0.09999999999999964) ellipse (1.9988885800753267 and 0.9994442900376633);
\fill[gray] (0, 0, 0.0) ellipse (2.0 and 1.0);
\fill[gray] (0, 0, 0.10000000000000009) ellipse (1.9988885800753267 and 0.9994442900376633);
\fill[gray] (0, 0, 0.20000000000000018) ellipse (1.9955506062794353 and 0.9977753031397176);
\fill[gray] (0, 0, 0.30000000000000027) ellipse (1.98997487421324 and 0.99498743710662);
\fill[gray] (0, 0, 0.40000000000000036) ellipse (1.9821424996424675 and 0.9910712498212337);
\fill[gray] (0, 0, 0.5) ellipse (1.9720265943665387 and 0.9860132971832694);
\fill[gray] (0, 0, 0.6000000000000001) ellipse (1.9595917942265426 and 0.9797958971132713);
\fill[gray] (0, 0, 0.7000000000000002) ellipse (1.9447936194419762 and 0.9723968097209881);
\fill[gray] (0, 0, 0.8000000000000003) ellipse (1.9275776393067947 and 0.9637888196533974);
\fill[gray] (0, 0, 0.9000000000000004) ellipse (1.9078784028338913 and 0.9539392014169457);
\fill[gray] (0, 0, 1.0) ellipse (1.8856180831641267 and 0.9428090415820634);
\fill[gray] (0, 0, 1.1000000000000005) ellipse (1.8607047649270483 and 0.9303523824635241);
\fill[gray] (0, 0, 1.2000000000000002) ellipse (1.833030277982336 and 0.916515138991168);
\fill[gray] (0, 0, 1.2999999999999998) ellipse (1.80246744461277 and 0.901233722306385);
\fill[gray] (0, 0, 1.4000000000000004) ellipse (1.7688665548562132 and 0.8844332774281066);
\fill[gray] (0, 0, 1.5) ellipse (1.7320508075688772 and 0.8660254037844386);
\fill[gray] (0, 0, 1.6000000000000005) ellipse (1.6918103387266024 and 0.8459051693633012);
\fill[gray] (0, 0, 1.7000000000000002) ellipse (1.6478942792411033 and 0.8239471396205517);
\fill[gray] (0, 0, 1.8000000000000007) ellipse (1.5999999999999996 and 0.7999999999999998);
\fill[gray] (0, 0, 1.9000000000000004) ellipse (1.5477582354991866 and 0.7738791177495933);
\fill[gray] (0, 0, 2.0) ellipse (1.4907119849998598 and 0.7453559924999299);
\fill[gray] (0, 0, 2.1000000000000005) ellipse (1.4282856857085697 and 0.7141428428542849);
\fill[gray] (0, 0, 2.2) ellipse (1.359738536958076 and 0.679869268479038);
\fill[gray] (0, 0, 2.3000000000000007) ellipse (1.2840906856172143 and 0.6420453428086071);
\fill[gray] (0, 0, 2.4000000000000004) ellipse (1.1999999999999997 and 0.5999999999999999);
\fill[gray] (0, 0, 2.5) ellipse (1.1055415967851334 and 0.5527707983925667);
\fill[gray] (0, 0, 2.6000000000000005) ellipse (0.9977753031397172 and 0.4988876515698586);
\fill[gray] (0, 0, 2.7) ellipse (0.8717797887081344 and 0.4358898943540672);
\fill[gray] (0, 0, 2.8000000000000007) ellipse (0.7180219742845992 and 0.3590109871422996);
\fill[gray] (0, 0, 2.9000000000000004) ellipse (0.5120763831912397 and 0.25603819159561986);
\fill[gray] (0, 0, 3.0) ellipse (0.0 and 0.0);
    \draw[gray!50!white] (0,0,3) -- (4,0,3);
    \draw[gray!50!white] (2,0,0) -- (2,3,0);
    \draw[gray!50!white] (0,1,0) -- (0,1,6);
    \draw[latex-latex, gray] (3.5,0,-3) -- (3.5,0,3);
    \draw[latex-latex, gray] (-2,2.5,0) -- (2,2.5,0);
    \draw[latex-latex, gray] (0,-1,5.5) -- (0,1,5.5);
    \draw[gray] (3.8,0,0) node {$a$};
    \draw[gray] (0,2.8,0) node {$b$};
    \draw[gray] (0,0,6) node {$c$};
    \begin{scope}
      \clip (-2,1) -- (2,-1) -- (2,1);
      \draw[dashed] ellipse (2 and 1);
    \end{scope}
    \begin{scope}[rotate around x=90]
      \clip (-2,3) -- (2,-3) -- (2,3);
      \draw[dashed] ellipse (2 and 3);
    \end{scope}
    \begin{scope}[rotate around y=-90]
      \clip (-3,1) -- (3,-1) -- (3,1);
      \draw[dashed] ellipse (3 and 1);
    \end{scope}
    \draw[-latex] (2,0,0) -- (3,0,0);
    \draw (3.1,0.1,0) node {$x$};
    \draw[-latex] (0,1,0) -- (0,2,0);
    \draw (0.1,2.1,0) node {$y$};
    \draw[-latex] (0,0,3) -- (0,0,5);
    \draw (0,0,5.2) node {$z$};
  \end{tikzpicture}\\
  \begin{tabular}{|p{2cm}|p{1.5cm}|p{1.5cm}|}
    \hline
    Zingg classes&$y^1<2/3$&$y^1>2/3$\\
    \hline
    $y^2>2/3$&Class I\newline DISC&Class II\newline SPHERE\\
    \hline
    $y^2<2/3$&Class III\newline BLADE&Class IV\newline ROD\\
    \hline
  \end{tabular}
  \caption{The Zingg catalog}\label{fig:zingg}
\end{figure}

Zingg's catalog, illustrated in Figure~\ref{fig:zingg} is finite (it defines 4 classes) and it is complete, as each class contains geometric shapes (in fact, each class contains ellipsoids). On the other hand, the Zingg classification is artificial, as the threshold $2/3$ is an arbitrary choice. One can, of course, study generalized Zingg catalogs where this threshold is being varied~\cite{szabo_pebbles}, however, there exits infinitely many generalized Zingg catalogs and it is not clear, which one should be used. Although artificial, the Zingg classification still has the advantage to offer some bias: in coastal environments \emph{blades} and \emph{discs} appear to be dominant, admitting conclusions about the effects of friction~\cite{Domokos_Gibbons_2019}. We remark that we also computed the Zingg classes for the laboratory dataset, shown in Online Resource 1.

\subsection{Mechanical shape descriptors and natural classes}\label{ss:mechanical}

\subsubsection{The primary mechanical classification}\label{sss:primary}

The \emph{primary mechanical classification}, introduced in~\cite{varkonyi_gomboc} is based on the number of different types of \emph{equilibrium points}, that is, positions where the body is at rest when supported on a horizontal surface, under gravity. When pushed gently from any direction, the body sitting on a \emph{stable} equilibrium returns to its original position, while it tips over from an \emph{unstable} equilibrium. We denote the respective numbers for stable and unstable equilibria by $S,U$ and we also note that 3\=/dimensional objects also have $H$~\emph{saddle}\=/type equilibria,  where the behavior depends on the direction of the push. The mathematical background of the equilibrium points is the analysis of the \emph{radial distance function}~$r_K$ measured from the center of gravity~$o$ of the body~$K$. For planar objects this is a function~$r_K(\varphi)$ of the single polar angle~$\varphi$ while in 3~dimensions we have $r_K(\varphi, \theta)$ depending on two angles. Stable, unstable and saddle points of the body correspond to the minima, maxima and saddles of this function, respectively. These concepts are illustrated in Figure~\ref{fig:distance} for an ellipse in two\=/dimensions and for an ellipsoid in three\=/dimensions. For the latter, with main axes $a>b>c$ we can use the following parameterization in the orientation depicted in Figure~\ref{fig:zingg}:
\begin{displaymath}
  \begin{aligned}
    x(\varphi,\theta)&=\tfrac{b}{2}\sin\theta\cos\varphi\\
    y(\varphi,\theta)&=\tfrac{c}{2}\sin\theta\sin\varphi\\
    z(\varphi,\theta)&=\tfrac{a}{2}\cos\theta
  \end{aligned}
\end{displaymath}
where $0\leq\varphi<2\pi$ and $0\leq\theta\leq\pi$. The distance from the ellipsoid's center of gravity~$o$ is then the Euclidean~distance:
\begin{equation}
  r_\text{ell}(\varphi,\theta)=
  \scriptstyle\sqrt{(x(\varphi,\theta)-o_x)^2+(y(\varphi,\theta)-o_y)^2+(z(\varphi,\theta)-o_z)^2}\label{eq:ellipsoid}
\end{equation}
The function $r_\text{ell}$ has maxima at $(0,0,\pm \frac{a}{2})$, saddles at $(\pm \frac{b}{2},0,0)$ and minima at $(0,\pm \frac{c}{2},0)$. We are going to return to this example in subsection~\ref{sec:equilibria} where we give a more rigorous description of the equilibria of convex surfaces defined by smooth functions. We are also giving a more precise definition of equilibria on convex polyhedra, another well studied subset of convex bodies later in Section~\ref{sec:polyhedral}. You can also see Figure~\ref{fig:poly_equilibria} with the equilibria of a regular tetrahedron marked.

\begin{figure}[!ht]
    \centering
    \includegraphics[width=\columnwidth]{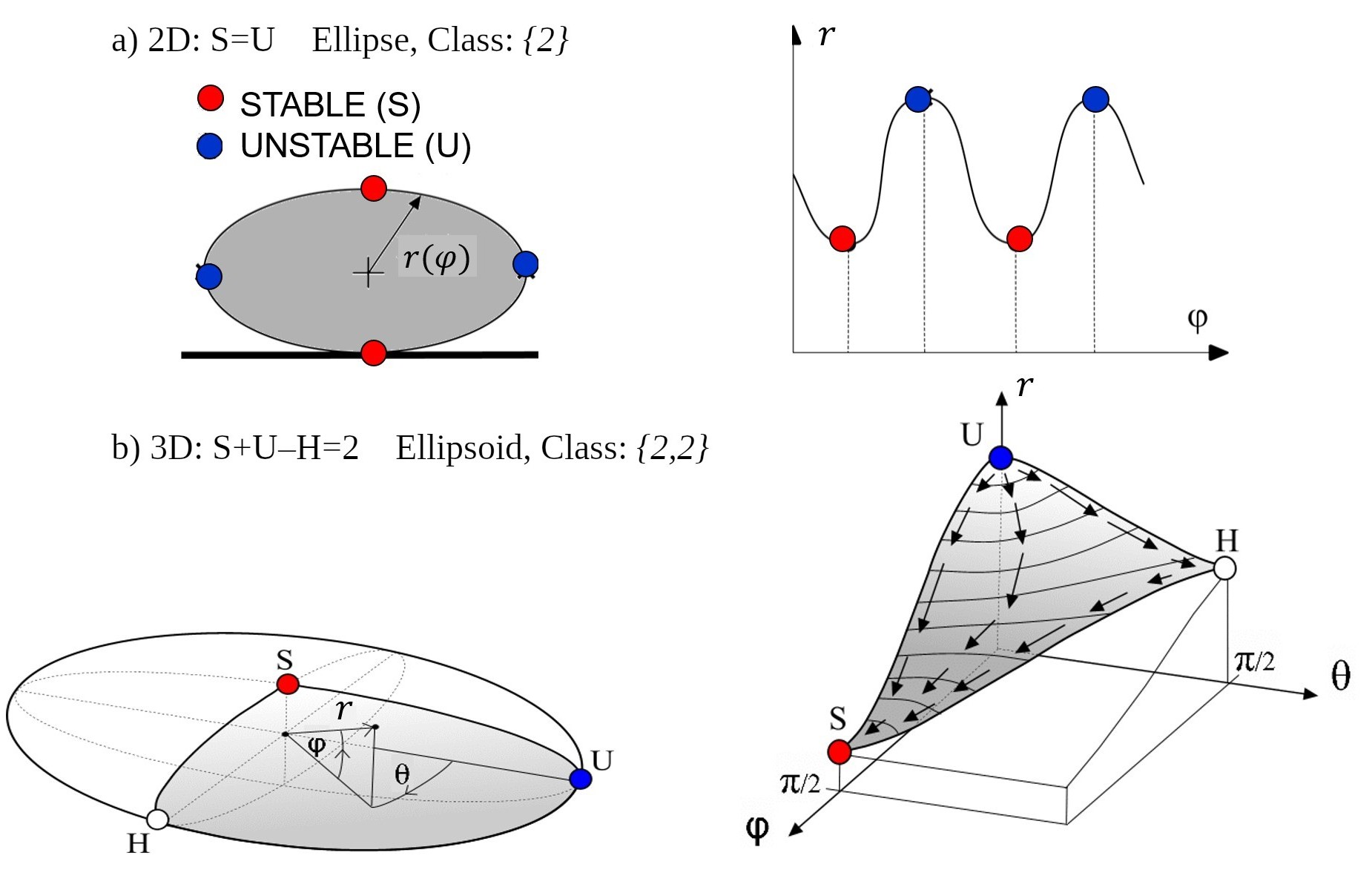}
    \caption{Mechanical equilibria as stationary points of the radial distance function $r$ a) In two\=/dimensions we have $r=r(\varphi)$ and stable and unstable equilibria appear alternating, in pairs, so we have $S=U$ b) In three\=/dimensions we have $r=r(\varphi, \theta)$. The respective numbers $S,U,H$ for stable, unstable and saddle\=/type equilibria satisfy the Poincaré\==Hopf formula $S+U-H=2$}\label{fig:distance}
\end{figure}

The numbers of different types of equilibria are related by the Poincaré\==Hopf formula $S+U-H=2$~\cite{milnor} so it is sufficient to record $S$ and $U$ and we call the pair $\{S,U\}$ the \emph{primary equilibrium class} of the body~\cite{varkonyi_gomboc, domokos_morse_smale, domokos_natural}. This means that the ellipsoid described by Eq.~\ref{eq:ellipsoid} is in the primary class \(\{2,2\}\). The primary mechanical classification system is illustrated in Figure~\ref{fig:primary}.

\begin{figure*}[!ht]
    \centering
    \includegraphics[width=1\textwidth]{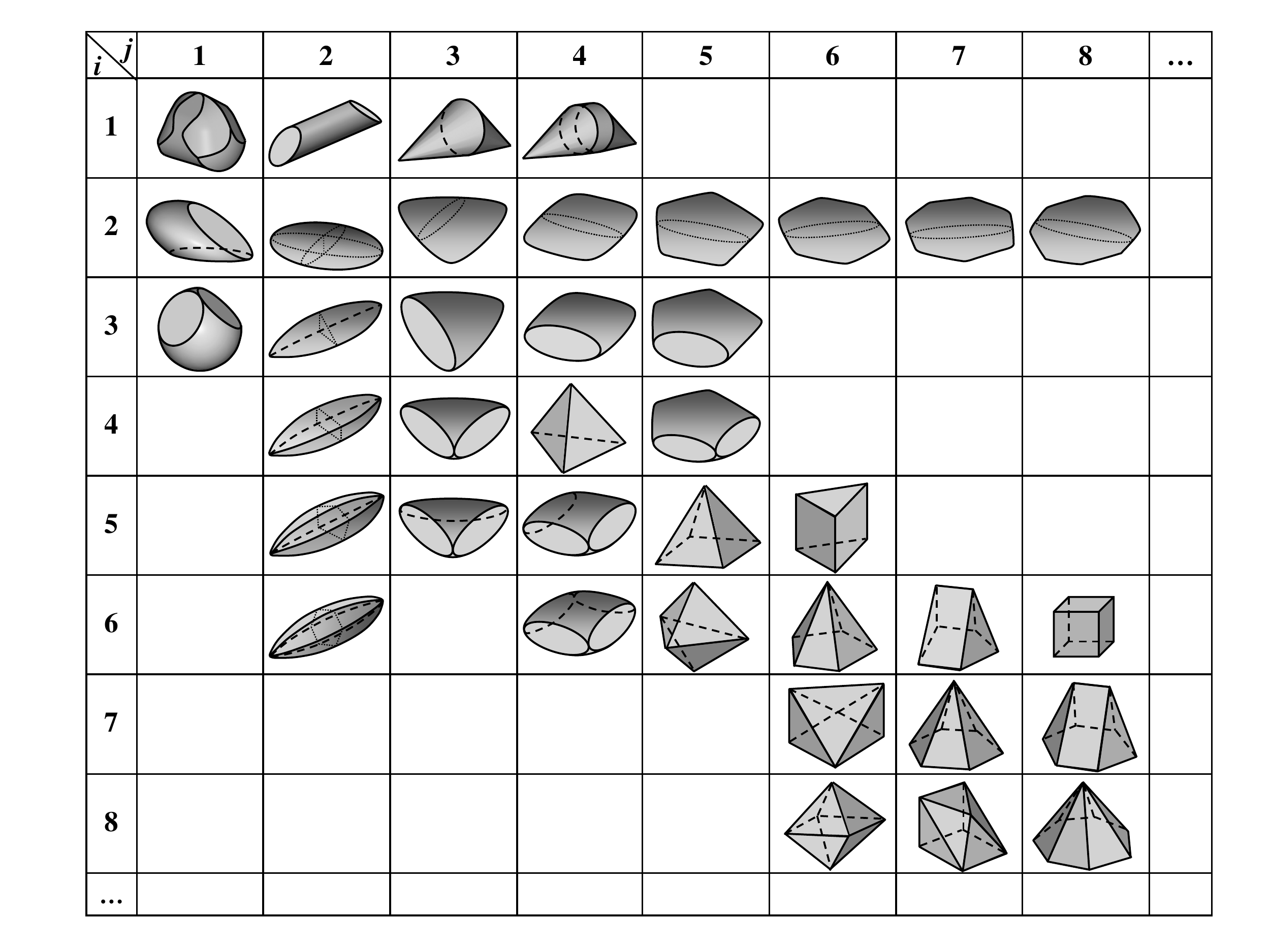}
    \caption{The primary mechanical catalog, showing examples of shapes with $i$ stable and $j$ unstable static balance points}\label{fig:primary}
\end{figure*}

Unlike the Zingg system, the primary mechanical catalog is infinite, and it is also natural, since we do not make any arbitrary choices when assigning the primary class $\{S,U\}$ to a particular shape, the class is \emph{encoded in the shape itself}. We also know that this catalog is complete, meaning no primary class is empty~\cite{varkonyi_gomboc}. This mathematical fact, however, is not related to the statistical distribution of natural particles: datasets of beach pebbles classified according to the primary mechanical catalog show very strong bias, as approximately 70\% of all beach pebbles appear in class $\{S,U\}=\{2,2\}$~\cite{szabo_pebbles}, making $\{2,2\}$ the dominant primary class. This fact, along with analysis of the statistical distribution is a key geophysical evidence supporting the theory that the total number $N=S+U+H$ of equilibria is monotonically decreasing in natural abrasion~\cite{domokos_natural}.

This bias is also present in the dataset with 271 pebbles we based our results on. Their distribution in the primary classes is presented in Table~\ref{tbl:distribution}. We will have special focus on the class \(\{2,2\}\) which contains 52 pebbles in our dataset.

\begin{table*}[!ht]
  \centering
  \caption{Pebbles used in the experiment}\label{tbl:distribution}
  \begin{subtable}{\textwidth}
    \centering
    \begin{tabular}{r|*{7}{r}|r}
      \diagbox{S}{U}&1&2&3&4&5&6&7&sum\\
      \hline
      1&0&0&0&0&0&0&0&0\\
      2&0&\cellcolor[gray]{0.8}52&\cellcolor[gray]{0.8}32&\cellcolor[gray]{0.8}7&\cellcolor[gray]{0.8}2&0&0&93\\
      3&0&\cellcolor[gray]{0.8}13&\cellcolor[gray]{0.8}28&\cellcolor[gray]{0.8}10&\cellcolor[gray]{0.8}3&1&0&55\\
      4&0&\cellcolor[gray]{0.8}15&\cellcolor[gray]{0.8}28&\cellcolor[gray]{0.8}19&\cellcolor[gray]{0.8}3&1&0&66\\
      5&0&\cellcolor[gray]{0.8}2&\cellcolor[gray]{0.8}10&\cellcolor[gray]{0.8}10&\cellcolor[gray]{0.8}12&1&1&36\\
      6&0&3&5&2&3&0&0&13\\
      7&0&0&0&3&1&1&1&6\\
      8&0&0&0&0&0&2&0&2\\
      \hline
      sum&0&85&103&51&24&6&2&271
    \end{tabular}
    \caption{Number of pebbles in each primary equilibrium class}
  \end{subtable}
  \begin{subtable}{\textwidth}
    \centering
    \begin{tabular}{r|@{}*{4}{c@{}}}
      \diagbox{S}{U}&2&3&4&5\\
      \hline
      2&
      \raisebox{-1cm}{\includegraphics[width=2cm]{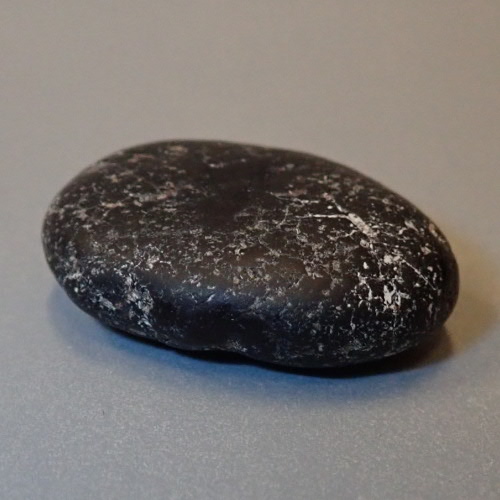}}&
      \raisebox{-1cm}{\includegraphics[width=2cm]{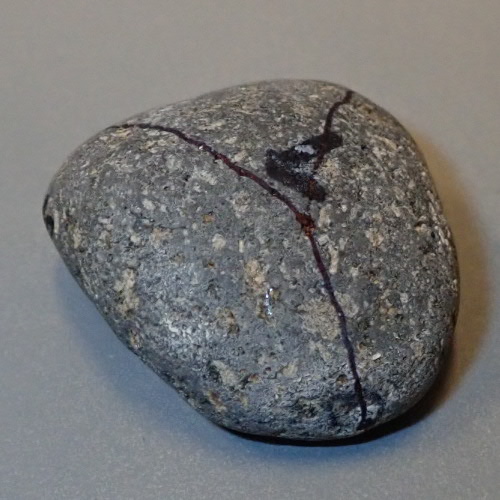}}&
      \raisebox{-1cm}{\includegraphics[width=2cm]{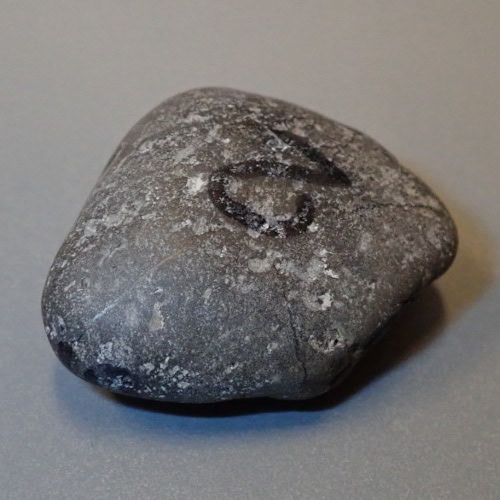}}&
      \raisebox{-1cm}{\includegraphics[width=2cm]{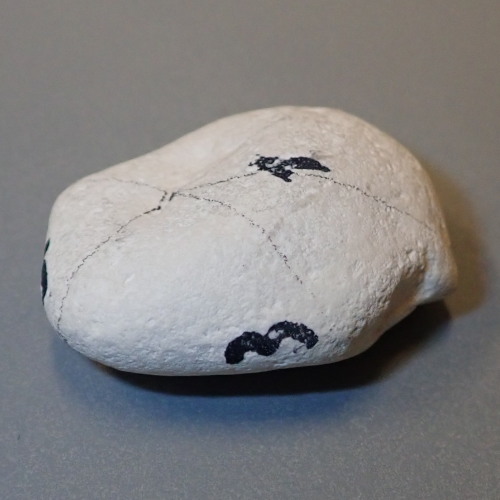}}\\
      3&
      \raisebox{-1cm}{\includegraphics[width=2cm]{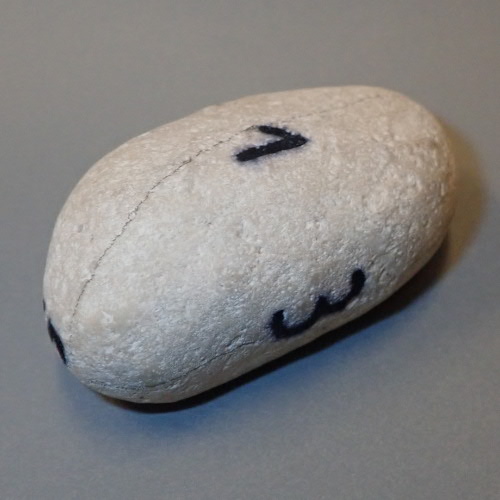}}&
      \raisebox{-1cm}{\includegraphics[width=2cm]{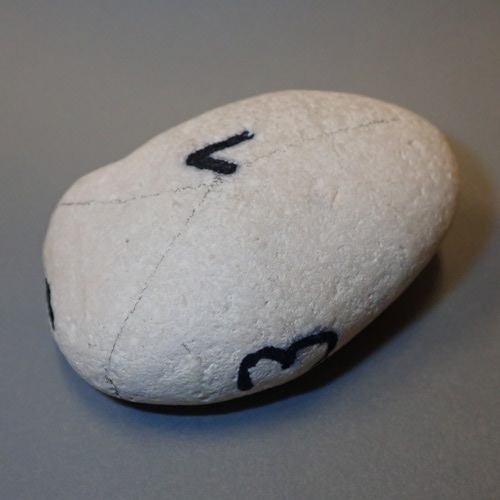}}&
      \raisebox{-1cm}{\includegraphics[width=2cm]{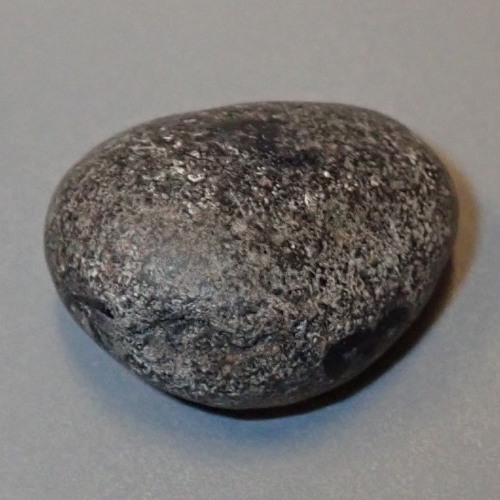}}&
      \raisebox{-1cm}{\includegraphics[width=2cm]{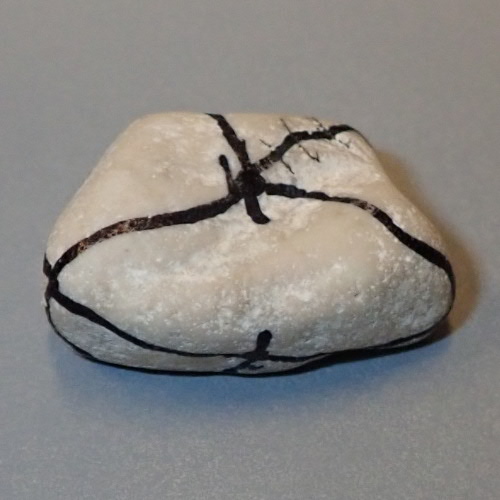}}\\
      4&
      \raisebox{-1cm}{\includegraphics[width=2cm]{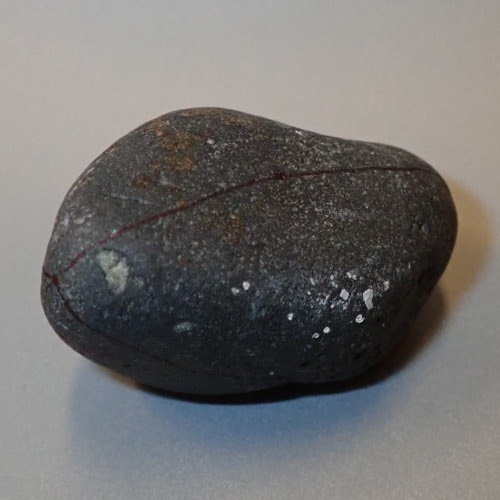}}&
      \raisebox{-1cm}{\includegraphics[width=2cm]{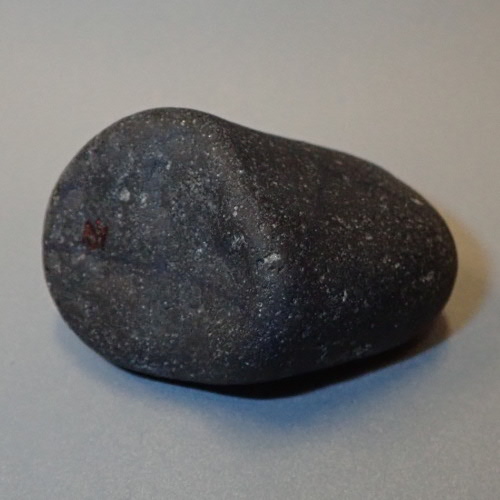}}&
      \raisebox{-1cm}{\includegraphics[width=2cm]{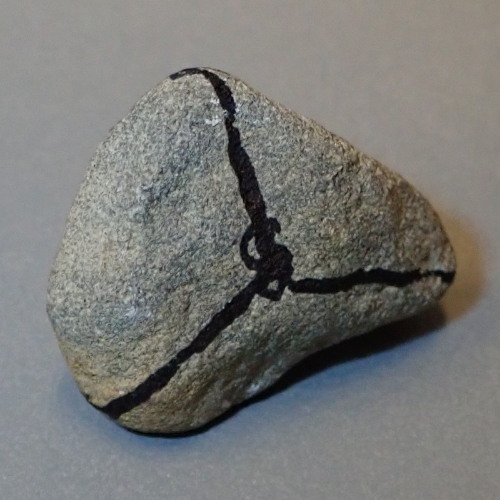}}&
      \raisebox{-1cm}{\includegraphics[width=2cm]{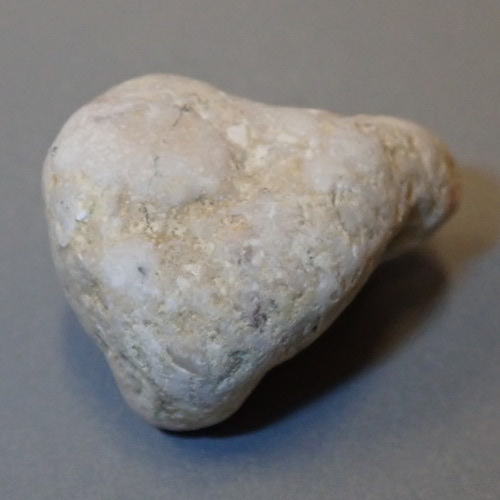}}\\
      5&
      \raisebox{-1cm}{\includegraphics[width=2cm]{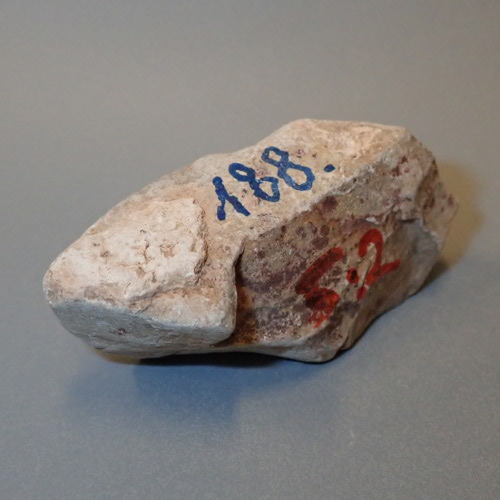}}&
      \raisebox{-1cm}{\includegraphics[width=2cm]{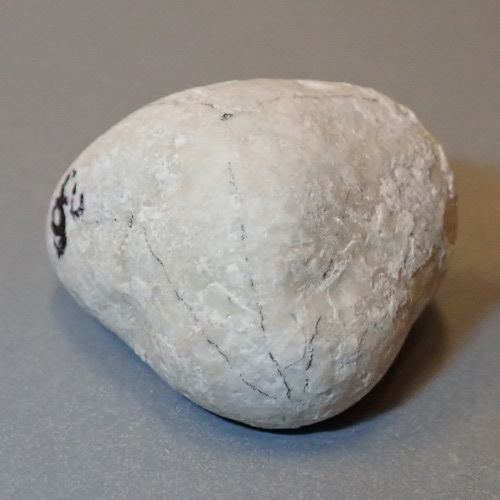}}&
      \raisebox{-1cm}{\includegraphics[width=2cm]{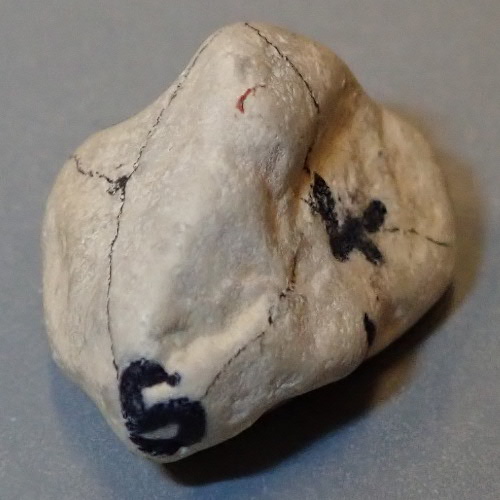}}&
      \raisebox{-1cm}{\includegraphics[width=2cm]{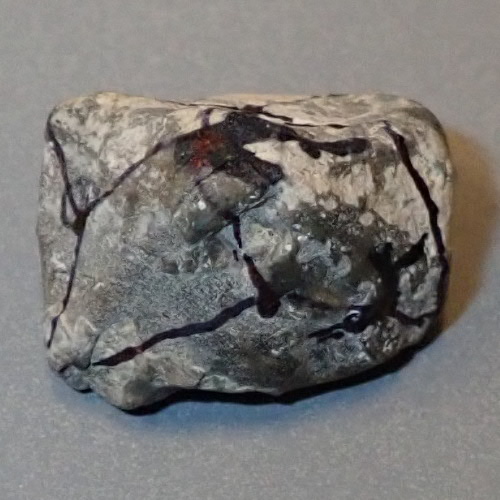}}
    \end{tabular}
    \caption{A pebble from each class in the highlighted region}
  \end{subtable}
\end{table*}

While the strong bias in the primary mechanical classification is most telling, it also has the drawback that it does not offer any clue on shapes \emph{inside} any of the primary classes, most notably, inside the dominant class $\{S,U\}=\{2,2\}$.  

\subsubsection{Higher order mechanical classifications}\label{sss:higher}

The trivial drawback of catalogs is that they do not distinguish between shapes in one class. This can be remedied by introducing finer, higher order catalogs. In case of mechanical descriptors this idea led to study, beyond the number, also the \emph{relative locations} of equilibrium points. We can have two alternative approaches to this task: we can either apply a natural, discrete decomposition to the \emph{range} of the distance function $r=r(\varphi, \theta)$, or, we can apply a natural, discrete decomposition of its \emph{domain}. To decompose the range, we will use saddle points and to decompose the domain, we will use isolated integral curves of the gradient. In the first case we arrive at a graph defining a natural \emph{hierarchy} among equilibria, in the second case we arrive at a graph defining a natural \emph{arrangement} among equilibria. We describe both graphs below.
 
The hierarchy among equilibria is based on the value of the radial distance function, and this information is carried by the \emph{Reeb\=/graph}~$R(K)$ associated with the body~$K$~\cite{arnold_counting}. Each point of an edge in the Reeb\=/graph~$R(K)$ corresponds to a connected component of the level set $r_K=\text{constant}$. We will define the Reeb\=/graph more rigorously for smooth functions in subsection~\ref{sec:reeb}, where we also show how the ellipsoid of Eq.~\ref{eq:ellipsoid} is degenerate due to its symmetry. A similar issue arises with the regular tetrahedron, an issue we will discuss in subsection~\ref{sec:poly_reeb} along with the definition of the Reeb\=/graph for polyhedra. Slightly moving the reference point~\(o\) off all 3~symmetry planes of the ellipsoid results in the non\=/degenerate Reeb\=/graph shown in Figure~\ref{fig:ell_reeb}. However, this is not the only Reeb\=/graph in the primary class \(\{2,2\}\). The number $R(S,U)$ of distinct Reeb\=/graphs in each primary class grows exponentially~\cite{arnold_counting, nicolaescu_counting}. Reeb\=/graphs define a finite, complete natural catalog inside any primary equilibrium class and of course they define an infinite catalog if we consider all primary classes. We call this the \emph{R\=/secondary mechanical classification} scheme. Whether and to what extent R\=/secondary classification is biased, we will discuss below.

The information on the spatial arrangement of equilibria can be described by using the isolated integral curves of the gradient vector field~$\nabla r_K$~\cite{guckenheimerholmes}, which define, as edges,  the so\=/called Morse\==Smale graph~$M(K)$ associated with the body~$K$~\cite{domokos_morse_smale}. The vertices of the Morse\==Smale graph are the static equilibrium points. Morse\==Smale graphs have various equivalent representations, to which we return in subsection~\ref{sec:ms} where we define the Morse\==Smale graph for smooth functions. Each of these representations can be enumerated and thus can be associated with an integer label, the so\=/called canonical code~\cite{babai}. A polyhedral version of the Morse\==Smale graph also exists, which well be discussed in subsection~\ref{sec:poly_ms}.

The Morse\==Smale graph belonging to the ellipsoid in the previous example is shown in Figure~\ref{fig:ell_ms}. However, this is not the only possible Morse\==Smale graph in the primary class $\{S,U\}=\{2,2\}$. The number~$M(S,U)$ of distinct Morse\==Smale graphs in each primary class grows approximately with the exponent $p = S + U$~\cite{kapolnai}. Similarly to Reeb\=/graphs, Morse\==Smale graphs define a finite, complete natural catalog inside any primary equilibrium class and they define an infinite catalog if we consider all primary classes. We call this the \emph{M\=/secondary mechanical classification} scheme. Similarly to Reeb\=/graphs, we will discuss whether the M\=/secondary classification is biased.
 
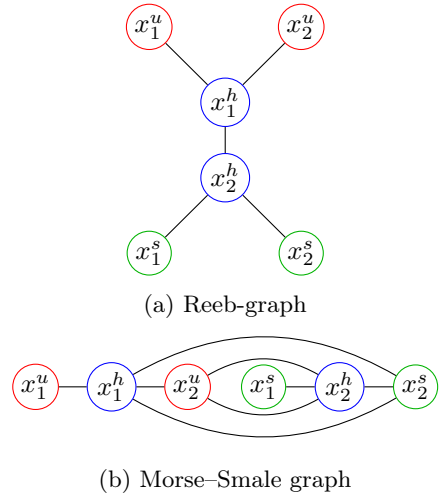
\begin{figure}[!ht]
  \begin{subfigure}{\columnwidth}
    \centering
    \begin{tikzpicture}[baseline=1.5cm]
      \node [circle, draw=red, fill=white, inner sep=1pt] (U1) at (0,3) {$x^u_1$};
      \node [circle, draw=red, fill=white, inner sep=1pt] (U2) at (2,3) {$x^u_2$};
      \node [circle, draw=blue, fill=white, inner sep=1pt] (H1) at (1,2) {$x^h_1$};
      \node [circle, draw=blue, fill=white, inner sep=1pt] (H2) at (1,1) {$x^h_2$};
      \node [circle, draw=green!70!black, fill=white, inner sep=1pt] (S1) at (0,0) {$x^s_1$};
      \node [circle, draw=green!70!black, fill=white, inner sep=1pt] (S2) at (2,0) {$x^s_2$};
      \draw (H1) to (U1) (H1) to (U2) (H2) to (H1) (S1) to (H2) (S2) to (H2);
    \end{tikzpicture}
    \caption{Reeb\=/graph}\label{fig:ell_reeb}
  \end{subfigure}\\
  \begin{subfigure}{\columnwidth}
    \centering
    \begin{tikzpicture}[baseline=-0.5]
      \node[circle, draw=red, fill=white, inner sep=1pt] (U1) at (0,0) {$x^u_1$};
      \node[circle, draw=blue, fill=white, inner sep=1pt] (H1) at (1,0) {$x^h_1$};
      \node[circle, draw=red, fill=white, inner sep=1pt] (U2) at (2,0) {$x^u_2$};
      \node[circle, draw=green!70!black, fill=white, inner sep=1pt] (S1) at (3,0) {$x^s_1$};
      \node[circle, draw=blue, fill=white, inner sep=1pt] (H2) at (4,0) {$x^h_2$};
      \node[circle, draw=green!70!black, fill=white, inner sep=1pt] (S2) at (5,0) {$x^s_2$};
      \draw (H1) to[bend left] (S2) (H1) to[bend right] (S2) (H1) to (U1) (H1) to (U2);
      \draw (H2) to (S1) (H2) to (S2) (H2) to[bend right] (U2) (H2) to[bend left] (U2);
    \end{tikzpicture}
    \caption{Morse\==Smale graph}\label{fig:ell_ms}
  \end{subfigure}
  \caption{Secondary mechanical descriptors associated with the tri\=/axial ellipsoid
  given in Eq.~\ref{eq:ellipsoid} (a) Reeb\=/graph with center of mass~\(o\) slightly offset from symmetry planes (b) Morse\==Smale graph, center of mass~\(o\) can be located either at the intersection of symmetry planes or with small offset}\label{fig:ms_and_reeb}
\end{figure}

These R\=/secondary and M\=/secondary  classification schemes are independent in the sense that identifying the class of a body in one system  does not locate it in the other system. To bridge the gap we will introduce the \emph{master graph}~$G(K)$ associated with the body~$K$, which carries information on both the integral curves and the level sets of the radial distance function and from which both the M\=/secondary and the R\=/secondary classification can be obtained. The master graph is going to be defined separately for smooth surfaces in subsection~\ref{sec:master} and for polyhedra in subsection~\ref{sec:poly_master}. We will call the scheme defined by the master graph the \emph{tertiary mechanical classification}, defining up to $M(S,U) \cdot R(S,U)$ tertiary classes in the primary class $\{S,U\}$. We will denote their exact number by \(G(S,U)\). Inside any primary class $\{S,U\}$, the tertiary scheme defines a natural, finite catalog. However, unlike any previous catalog, the tertiary classification is incomplete as there exists pairs of Reeb\=/graphs and Morse\==Smale graphs in the same primary class which can not belong to the same object. We will illustrate the tertiary classification scheme on the primary class $\{S,U\}=\{2,2\}$ and explain its significance. In particular, in Section~\ref{sec:catalog} we will prove

\begin{lemma}\label{lem:cardinality}
  \begin{displaymath}
    G(2,2)=3
  \end{displaymath}
\end{lemma}

For an example see Figure~\ref{fig:subdivision} showing the possible secondary and tertiary classes in the primary class with 2~stable and 2~unstable points. A more detailed view of this primary class and its higher order subclasses will be presented later in Table~\ref{tbl:twotwo}.

\begin{figure*}[!ht]
  \centering
  \begin{tikzpicture}
    \colorlet{stable}{green!70!black};
    \colorlet{saddle}{blue};
    \colorlet{unstable}{red};
    \colorlet{inter}{orange};
    \draw[step=1cm, gray] (0,0) grid (4,-4);
    \foreach \s in {1,...,4} {
      \foreach \u in {1,...,4} {
        \draw (\u-0.5,-\s+0.5) node {$\{\s,\u\}$};
      }
    }
    \draw[ultra thick] (1,-1) rectangle (2,-2);
    \coordinate (p-top-right) at (2,-1);
    \coordinate (p-bottom-right) at (2,-2);

    \begin{scope}[xshift=6.5cm, rotate=-22.5]
      \draw[gray] (-0.5,1) coordinate (r-top-left) --
      (0.5,1) coordinate (r-top-right) --
      (0.5,-1) coordinate (r-bottom-right) --
      (-0.5,-1) coordinate (r-bottom-left) --
      cycle;
      \draw[gray] (-0.5,0) -- (0.5,0);
      
      \draw[thin] (-0.25,0.8) -- (0.25,0.4);
      \draw[thin] (-0.25,0.6) -- (0.25,0.2);
      \draw[thin] (0,0.6) -- (0,0.4);
      \fill[unstable] (-0.25,0.8) circle (0.05);
      \fill[unstable] (-0.25,0.6) circle (0.05);
      \fill[saddle] (0,0.6) circle (0.05);
      \fill[saddle] (0,0.4) circle (0.05);
      \fill[stable] (0.25,0.4) circle (0.05);
      \fill[stable] (0.25,0.2) circle (0.05);
      
      \draw[thin] (-0.25,-0.2) -- (0,-0.4) -- (0.25,-0.2);
      \draw[thin] (-0.25,-0.8) -- (0,-0.6) -- (0.25,-0.8);
      \draw[thin] (0,-0.4) -- (0,-0.6);
      \fill[unstable] (-0.25,-0.2) circle (0.05);
      \fill[unstable] (0.25,-0.2) circle (0.05);
      \fill[saddle] (0,-0.4) circle (0.05);
      \fill[saddle] (0,-0.6) circle (0.05);
      \fill[stable] (-0.25,-0.8) circle (0.05);
      \fill[stable] (0.25,-0.8) circle (0.05);

      \draw[ultra thick] (-0.5,0) rectangle (0.5,-1);
    \end{scope}

    \begin{scope}[xshift=6.5cm, yshift=-3cm, rotate=22.5]
      \draw[gray] (-0.5,1) coordinate (m-top-left) --
      (0.5,1) coordinate (m-top-right) --
      (0.5,-1) coordinate (m-bottom-right) --
      (-0.5,-1) coordinate (m-bottom-left) --
      cycle;
      \draw[gray] (-0.5,0) -- (0.5,0);

      \draw[thin] (0,0.9) -- (0,0.58);
      \draw[thin] (0,0.42) -- (0,0.1);
      \draw[thin, bend left=90] (0,0.74) to (0,0.1);
      \draw[thin, bend left=90] (0,0.58) to (0,0.26);
      \draw[thin, bend right=90] (0,0.9) to (0,0.26);
      \draw[thin, bend right=90] (0,0.74) to (0,0.42);
      \fill[stable] (0,0.9) circle (0.05);
      \fill[saddle] (0,0.74) circle (0.05);
      \fill[stable] (0,0.58) circle (0.05);
      \fill[unstable] (0,0.42) circle (0.05);
      \fill[saddle] (0,0.26) circle (0.05);
      \fill[unstable] (0,0.1) circle (0.05);

      \draw[thin] (0,-0.1) -- (0,-0.42);
      \draw[thin] (0,-0.58) -- (0,-0.9);
      \draw[thin, bend right=90] (0,-0.1) to (0,-0.74);
      \draw[thin, bend right=90] (0,-0.26) to (0,-0.58);
      \draw[thin, bend left=90] (0,-0.1) to (0,-0.74);
      \draw[thin, bend left=90] (0,-0.26) to (0,-0.58);
      \fill[stable] (0,-0.1) circle (0.05);
      \fill[saddle] (0,-0.26) circle (0.05);
      \fill[stable] (0,-0.42) circle (0.05);
      \fill[unstable] (0,-0.58) circle (0.05);
      \fill[saddle] (0,-0.74) circle (0.05);
      \fill[unstable] (0,-0.9) circle (0.05);

      \draw[ultra thick] (-0.5,0) rectangle (0.5,1);
    \end{scope}
    \begin{scope}[xshift=10.5cm, yshift=-1.5cm, rotate=45]
      \draw[gray] (-1,1) coordinate (t-top-left) --
      (1,1) coordinate (t-top-right) --
      (1,-1) coordinate (t-bottom-right) --
      (-1,-1) coordinate (t-bottom-left) --
      cycle;
      \draw[gray] (-1,0) -- (1,0) (0,-1) -- (0,1);
      \draw (0.5,0.5) node {$\emptyset$};

      \draw[thin] (-0.5,0.9) -- (-0.5,0.58);
      \draw[thin] (-0.5,0.42) -- (-0.5,0.1);
      \draw[thin, bend left=90] (-0.5,0.74) to (-0.5,0.1);
      \draw[thin, bend left=90] (-0.5,0.58) to (-0.5,0.26);
      \draw[thin, bend right=90] (-0.5,0.9) to (-0.5,0.26);
      \draw[thin, bend right=90] (-0.5,0.74) to (-0.5,0.42);
      \draw[thin, gray] (-0.67,0.74) -- (-0.5,0.74);
      \draw[thin, gray, bend right=45] (-0.6,0.58) to (-0.5,0.26);
      \draw[thin, gray, bend right=45] (-0.4,0.42) to (-0.5,0.74);
      \draw[thin, gray] (-0.33,0.26) -- (-0.5,0.26);
      \fill[stable] (-0.5,0.9) circle (0.05);
      \fill[saddle] (-0.5,0.74) circle (0.05);
      \fill[stable] (-0.5,0.58) circle (0.05);
      \fill[unstable] (-0.5,0.42) circle (0.05);
      \fill[saddle] (-0.5,0.26) circle (0.05);
      \fill[unstable] (-0.5,0.1) circle (0.05);
      \fill[inter] (-0.67,0.74) circle (0.04);
      \fill[inter] (-0.6,0.58) circle (0.04);
      \fill[inter] (-0.4,0.42) circle (0.04);
      \fill[inter] (-0.33,0.26) circle (0.04);

      \begin{scope}[yshift=-1cm]
        \draw[thin] (-0.5,0.9) -- (-0.5,0.58);
        \draw[thin] (-0.5,0.42) -- (-0.5,0.1);
        \draw[thin, bend right=90] (-0.5,0.26) to (-0.5,0.9);
        \draw[thin, bend right=90] (-0.5,0.42) to (-0.5,0.74);
        \draw[thin, bend right=90] (-0.5,0.9) to (-0.5,0.26);
        \draw[thin, bend right=90] (-0.5,0.74) to (-0.5,0.42);
        \draw[thin, gray] (-0.67,0.74) -- (-0.33,0.74);
        \draw[thin, gray, bend right=45] (-0.6,0.58) to (-0.5,0.26);
        \draw[thin, gray, bend left=45] (-0.4,0.58) to (-0.5,0.26);
        \fill[stable] (-0.5,0.9) circle (0.05);
        \fill[saddle] (-0.5,0.74) circle (0.05);
        \fill[stable] (-0.5,0.58) circle (0.05);
        \fill[unstable] (-0.5,0.42) circle (0.05);
        \fill[saddle] (-0.5,0.26) circle (0.05);
        \fill[unstable] (-0.5,0.1) circle (0.05);
        \fill[inter] (-0.67,0.74) circle (0.04);
        \fill[inter] (-0.6,0.58) circle (0.04);
        \fill[inter] (-0.4,0.58) circle (0.04);
        \fill[inter] (-0.33,0.74) circle (0.04);
      \end{scope}

      \draw[thin] (0.5,-0.1) -- (0.5,-0.42);
      \draw[thin] (0.5,-0.58) -- (0.5,-0.9);
      \draw[thin, bend right=90] (0.5,-0.1) to (0.5,-0.74) to cycle;
      \draw[thin, bend right=90] (0.5,-0.42) to (0.5,-0.58) to cycle;
      \draw[thin, gray, bend right=45] (0.5,-0.18) to (0.5,-0.74);
      \draw[thin, gray, bend left=45] (0.5,-0.26) to (0.5,-0.66);
      \fill[stable] (0.5,-0.1) circle (0.05);
      \fill[saddle] (0.5,-0.26) circle (0.05);
      \fill[stable] (0.5,-0.42) circle (0.05);
      \fill[unstable] (0.5,-0.58) circle (0.05);
      \fill[saddle] (0.5,-0.74) circle (0.05);
      \fill[unstable] (0.5,-0.9) circle (0.05);
      \fill[inter] (0.5,-0.18) circle (0.04);
      \fill[inter] (0.5,-0.66) circle (0.04);

      \draw[ultra thick] (-1,0) rectangle (0,1);
    \end{scope}
    \begin{scope}[dashed]
      \draw (p-top-right) -- (r-top-left);
      \draw (p-bottom-right) -- (r-bottom-left);
      \draw (p-top-right) -- (m-top-left);
      \draw (p-bottom-right) -- (m-bottom-left);
      \draw (r-top-right) -- (t-top-right);
      \draw (r-bottom-right) -- (t-top-left);
      \draw (m-top-right) -- (t-top-left);
      \draw (m-bottom-right) -- (t-bottom-left);
    \end{scope}
    \draw (2,1.5) node {primary};
    \draw (6.5,1.5) node {R\=/secondary};
    \draw (6.5,-4.5) node {M\=/secondary};
    \draw (10.5,1.5) node {tertiary};
    \begin{scope}[very thick, gray]
      \draw (4.5,2) -- (4.5,-5);
      \draw (8.5,2) -- (8.5,-5);
      \draw (4.5,-1.5) -- (8.5,-1.5);
    \end{scope}
  \end{tikzpicture}
  \caption{Division of the primary class \(\{2,2\}\) into higher order subclasses. The class of the ellipsoid used in previous examples is highlighted under each classification scheme}\label{fig:subdivision}
\end{figure*}
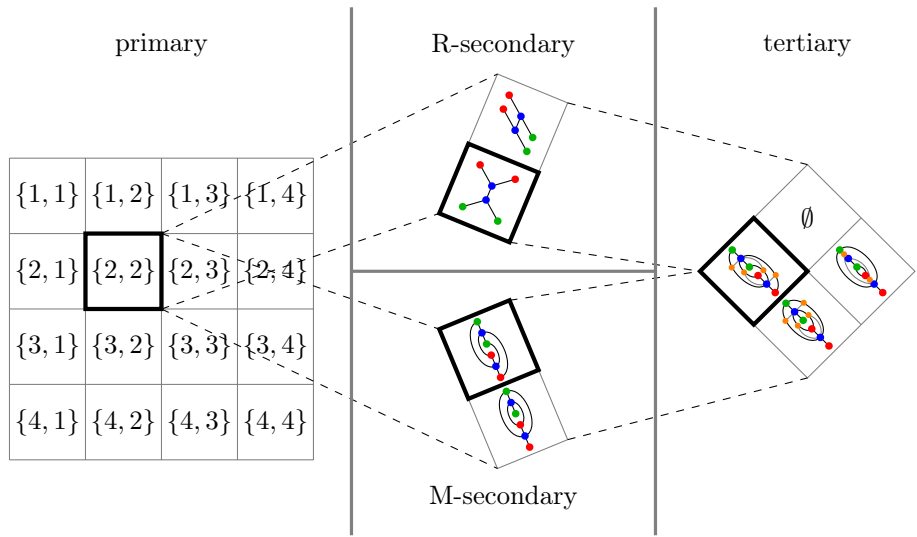

While these higher order classifications are most promising, the main obstacle in their application was that they cannot be reliably identified by hand experiments and there was no reliable algorithm and no computer\=/based tool to extract them from natural shapes. The first main result presented in this article is an algorithm computing all three previously mentioned graphs using the 3D scan of the particle.

We have investigated the primary class with the highest number of naturally occurring pebbles, the class of bodies with 2 stable and 2 unstable points. Every such pebble in our collection fell into the same tertiary class. This suggests that the natural abrasion process is heavily skewed towards certain classes, just like in the case of the primary classes.

\section{Classification of smooth, convex shapes}\label{sec:classification}

\subsection{Primary equilibrium classes}\label{sec:equilibria}

In subsection~\ref{ss:mechanical} we have introduced the radial distance function~$r_K(\varphi,\theta):S^2\to\mathbb{R}$, and also stated how its extrema correspond to equilibria of the convex body $K$. Next we examine the case where $r_K$ is at least twice continuously differentiable. In this case the non\=/degenerate \emph{critical point~$x\in S^2$} of the gradient field $\nabla r_K$ corresponds to a static equilibrium point of the body: a minimum, maximum or saddle point corresponds to a stable, unstable or saddle\=/type equilibrium point, respectively~\cite{varkonyi_gomboc}. Non\=/critical points are called \emph{regular}. We denote the type of equilibrium point with a superscript when relevant, $x^s, x^u$ and $x^h\in S^2$ are stable, unstable and saddle points, respectively. The capital letters $S, U$ and $H\in\mathbb{N}$ still represent the total number of these equilibria for $K$.

The smooth function~$r_K$ is called \emph{Morse} if all its critical points are non\=/degenerate. In this case the numbers of different equilibria are all finite because the function's domain is compact. The Poincaré\==Hopf theorem establishes the relationship between these values: $H=S+U-2$. If the \emph{critical value} $r_K(x)$ for every critical point~$x$ is distinct too then the function is called an \emph{excellent Morse function}~\cite{nicolaescu_counting}.

For example, let us return to the ellipsoid introduced in Eq.~\ref{eq:ellipsoid}. All of its equilibria are non\=/degenerate, but the ones of the same type are at equal distances from the center of gravity. This makes it a Morse function but not an excellent one. Choosing the value of $o$ such that none of its coordinates is zero breaks the symmetry, making $r_\text{ell}$ an excellent Morse function. Figure~\ref{fig:ell_equilibria} shows the visible equilibrium points for $o=(0.2,0.1,0.3)$. An additional minimum, maximum and saddle point is hidden behind the ellipsoid.

\begin{figure}[!ht]
  \centering
  \begin{tikzpicture}
    \fill[gray] (0, 0, -3.0) ellipse (0.0 and 0.0);
\fill[gray] (0, 0, -2.9) ellipse (0.5120763831912405 and 0.25603819159562025);
\fill[gray] (0, 0, -2.8) ellipse (0.7180219742846008 and 0.3590109871423004);
\fill[gray] (0, 0, -2.7) ellipse (0.8717797887081344 and 0.4358898943540672);
\fill[gray] (0, 0, -2.6) ellipse (0.9977753031397176 and 0.4988876515698588);
\fill[gray] (0, 0, -2.5) ellipse (1.1055415967851334 and 0.5527707983925667);
\fill[gray] (0, 0, -2.4) ellipse (1.2000000000000002 and 0.6000000000000001);
\fill[gray] (0, 0, -2.3) ellipse (1.284090685617215 and 0.6420453428086075);
\fill[gray] (0, 0, -2.2) ellipse (1.359738536958076 and 0.679869268479038);
\fill[gray] (0, 0, -2.1) ellipse (1.42828568570857 and 0.714142842854285);
\fill[gray] (0, 0, -2.0) ellipse (1.4907119849998598 and 0.7453559924999299);
\fill[gray] (0, 0, -1.9) ellipse (1.5477582354991866 and 0.7738791177495933);
\fill[gray] (0, 0, -1.7999999999999998) ellipse (1.6 and 0.8);
\fill[gray] (0, 0, -1.7) ellipse (1.6478942792411033 and 0.8239471396205517);
\fill[gray] (0, 0, -1.5999999999999999) ellipse (1.6918103387266028 and 0.8459051693633014);
\fill[gray] (0, 0, -1.5) ellipse (1.7320508075688772 and 0.8660254037844386);
\fill[gray] (0, 0, -1.4) ellipse (1.7688665548562132 and 0.8844332774281066);
\fill[gray] (0, 0, -1.2999999999999998) ellipse (1.80246744461277 and 0.901233722306385);
\fill[gray] (0, 0, -1.2) ellipse (1.8330302779823362 and 0.9165151389911681);
\fill[gray] (0, 0, -1.0999999999999999) ellipse (1.8607047649270483 and 0.9303523824635241);
\fill[gray] (0, 0, -1.0) ellipse (1.8856180831641267 and 0.9428090415820634);
\fill[gray] (0, 0, -0.8999999999999999) ellipse (1.9078784028338913 and 0.9539392014169457);
\fill[gray] (0, 0, -0.7999999999999998) ellipse (1.9275776393067947 and 0.9637888196533974);
\fill[gray] (0, 0, -0.6999999999999997) ellipse (1.9447936194419762 and 0.9723968097209881);
\fill[gray] (0, 0, -0.5999999999999996) ellipse (1.9595917942265426 and 0.9797958971132713);
\fill[gray] (0, 0, -0.5) ellipse (1.9720265943665387 and 0.9860132971832694);
\fill[gray] (0, 0, -0.3999999999999999) ellipse (1.9821424996424675 and 0.9910712498212337);
\fill[gray] (0, 0, -0.2999999999999998) ellipse (1.98997487421324 and 0.99498743710662);
\fill[gray] (0, 0, -0.19999999999999973) ellipse (1.9955506062794355 and 0.9977753031397177);
\fill[gray] (0, 0, -0.09999999999999964) ellipse (1.9988885800753267 and 0.9994442900376633);
\fill[gray] (0, 0, 0.0) ellipse (2.0 and 1.0);
\fill[gray] (0, 0, 0.10000000000000009) ellipse (1.9988885800753267 and 0.9994442900376633);
\fill[gray] (0, 0, 0.20000000000000018) ellipse (1.9955506062794353 and 0.9977753031397176);
\fill[gray] (0, 0, 0.30000000000000027) ellipse (1.98997487421324 and 0.99498743710662);
\fill[gray] (0, 0, 0.40000000000000036) ellipse (1.9821424996424675 and 0.9910712498212337);
\fill[gray] (0, 0, 0.5) ellipse (1.9720265943665387 and 0.9860132971832694);
\fill[gray] (0, 0, 0.6000000000000001) ellipse (1.9595917942265426 and 0.9797958971132713);
\fill[gray] (0, 0, 0.7000000000000002) ellipse (1.9447936194419762 and 0.9723968097209881);
\fill[gray] (0, 0, 0.8000000000000003) ellipse (1.9275776393067947 and 0.9637888196533974);
\fill[gray] (0, 0, 0.9000000000000004) ellipse (1.9078784028338913 and 0.9539392014169457);
\fill[gray] (0, 0, 1.0) ellipse (1.8856180831641267 and 0.9428090415820634);
\fill[gray] (0, 0, 1.1000000000000005) ellipse (1.8607047649270483 and 0.9303523824635241);
\fill[gray] (0, 0, 1.2000000000000002) ellipse (1.833030277982336 and 0.916515138991168);
\fill[gray] (0, 0, 1.2999999999999998) ellipse (1.80246744461277 and 0.901233722306385);
\fill[gray] (0, 0, 1.4000000000000004) ellipse (1.7688665548562132 and 0.8844332774281066);
\fill[gray] (0, 0, 1.5) ellipse (1.7320508075688772 and 0.8660254037844386);
\fill[gray] (0, 0, 1.6000000000000005) ellipse (1.6918103387266024 and 0.8459051693633012);
\fill[gray] (0, 0, 1.7000000000000002) ellipse (1.6478942792411033 and 0.8239471396205517);
\fill[gray] (0, 0, 1.8000000000000007) ellipse (1.5999999999999996 and 0.7999999999999998);
\fill[gray] (0, 0, 1.9000000000000004) ellipse (1.5477582354991866 and 0.7738791177495933);
\fill[gray] (0, 0, 2.0) ellipse (1.4907119849998598 and 0.7453559924999299);
\fill[gray] (0, 0, 2.1000000000000005) ellipse (1.4282856857085697 and 0.7141428428542849);
\fill[gray] (0, 0, 2.2) ellipse (1.359738536958076 and 0.679869268479038);
\fill[gray] (0, 0, 2.3000000000000007) ellipse (1.2840906856172143 and 0.6420453428086071);
\fill[gray] (0, 0, 2.4000000000000004) ellipse (1.1999999999999997 and 0.5999999999999999);
\fill[gray] (0, 0, 2.5) ellipse (1.1055415967851334 and 0.5527707983925667);
\fill[gray] (0, 0, 2.6000000000000005) ellipse (0.9977753031397172 and 0.4988876515698586);
\fill[gray] (0, 0, 2.7) ellipse (0.8717797887081344 and 0.4358898943540672);
\fill[gray] (0, 0, 2.8000000000000007) ellipse (0.7180219742845992 and 0.3590109871422996);
\fill[gray] (0, 0, 2.9000000000000004) ellipse (0.5120763831912397 and 0.25603819159561986);
\fill[gray] (0, 0, 3.0) ellipse (0.0 and 0.0);
    \input{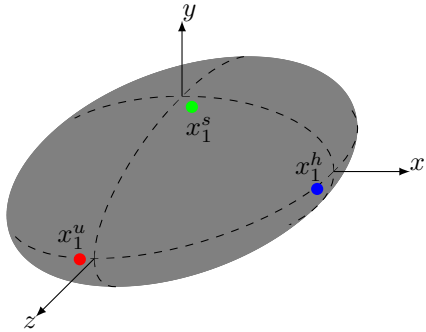}
    \begin{scope}
      \clip (-2,1) -- (2,-1) -- (2,1);
      \draw[dashed] ellipse (2 and 1);
    \end{scope}
    \begin{scope}[rotate around x=90]
      \clip (-2,3) -- (2,-3) -- (2,3);
      \draw[dashed] ellipse (2 and 3);
    \end{scope}
    \begin{scope}[rotate around y=-90]
      \clip (-3,1) -- (3,-1) -- (3,1);
      \draw[dashed] ellipse (3 and 1);
    \end{scope}
    \fill[green] (s2) circle (0.08);
    \draw ($(s2)+(0.1,-0.3)$) node {$x^s_1$};
    \fill[blue] (h2) circle (0.08);
    \draw ($(h2)+(-0.1,0.3)$) node {$x^h_1$};
    \fill[red] (u2) circle (0.08);
    \draw ($(u2)+(-0.1,0.3)$) node {$x^u_1$};
    \draw[-latex] (2,0,0) -- (3,0,0);
    \draw (3.1,0.1,0) node {$x$};
    \draw[-latex] (0,1,0) -- (0,2,0);
    \draw (0.1,2.1,0) node {$y$};
    \draw[-latex] (0,0,3) -- (0,0,5);
    \draw (0,0,5.2) node {$z$};
  \end{tikzpicture}
  \caption{All visible equilibria}\label{fig:ell_equilibria}
\end{figure}

The \emph{primary equilibrium class $\{i,j\}$} ($i,j=0,1,\dots$) contains all convex bodies with $S=i$ stable equilibria and $U=j$ unstable equilibria. Bodies with 1 stable point are called \emph{mono\=/stable}, bodies with 1 unstable point are called \emph{mono\=/unstable} and bodies in the class $\{1,1\}$ are called \emph{mono\=/monostatic}. Classifying an object whose radial distance function is Morse is trivial, but we will show in the next section that primary equilibrium classes are not limited to these.

Now we define two important tools~--~introduced in subsection~\ref{sss:higher} as second order mechanical descriptors~--~for Morse functions specifically, the Reeb\=/graph and the Morse\==Smale graph.

\subsection{Reeb\=/graph}\label{sec:reeb}

A \emph{level set} of a real valued function is the set of points where the function takes on a particular value ${d\in\mathbb{R}}$. In our case these are the points at distance~$d$ from the point~$o$. A level set of $r_K$ may consist of several closed curves, called \emph{contour lines}. Contour lines are unique at regular points. Contour lines at a stable or an unstable point consist of that single point. We will call a contour line of $r_K$ containing a saddle point~$x^s$ the \emph{saddle contour line of~$x^s$}. An excellent Morse function's saddle contour lines always contain a single saddle point due to the critical values being distinct. Therefore the saddle contour line of~$x^s$ has a single self\=/intersection at $x^s$, as shown in Figure~\ref{fig:ell_contour}.

\begin{figure}[!ht]
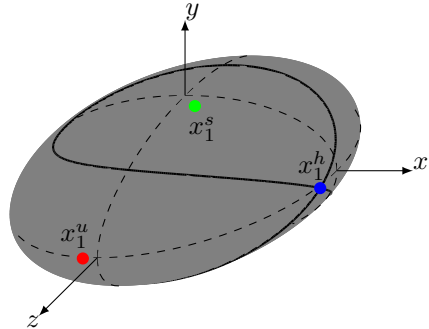

  \centering
  \begin{tikzpicture}
    \input{contours}
    \input{equilibria}
    \begin{scope}
      \clip (-2,1) -- (2,-1) -- (2,1);
      \draw[dashed] ellipse (2 and 1);
    \end{scope}
    \begin{scope}[rotate around x=90]
      \clip (-2,3) -- (2,-3) -- (2,3);
      \draw[dashed] ellipse (2 and 3);
    \end{scope}
    \begin{scope}[rotate around y=-90]
      \clip (-3,1) -- (3,-1) -- (3,1);
      \draw[dashed] ellipse (3 and 1);
    \end{scope}
    \fill[green] (s2) circle (0.08);
    \draw ($(s2)+(0.1,-0.3)$) node {$x^s_1$};
    \fill[blue] (h2) circle (0.08);
    \draw ($(h2)+(-0.1,0.3)$) node {$x^h_1$};
    \fill[red] (u2) circle (0.08);
    \draw ($(u2)+(-0.1,0.3)$) node {$x^u_1$};
    \draw[-latex] (2,0,0) -- (3,0,0);
    \draw (3.1,0.1,0) node {$x$};
    \draw[-latex] (0,1,0) -- (0,2,0);
    \draw (0.1,2.1,0) node {$y$};
    \draw[-latex] (0,0,3) -- (0,0,5);
    \draw (0,0,5.2) node {$z$};
  \end{tikzpicture}
  \caption{The saddle contour line of $x^h_1$}\label{fig:ell_contour}
\end{figure}

The number of contour lines is not constant throughout the codomain of $r_K$. New contour lines appear at critical values of minima, disappear at critical values of maxima, and split or merge at critical values of saddles. To capture this evolution of level sets, first an equivalence relation~$\sim_{r_K}$ is defined for two points $p, q\in S^2$ such that $p\sim_{r_K} q$ whenever $p$ and $q$ belong to the same contour line of $r_K$. The \emph{Reeb\=/graph of $r_K$} is then the quotient space $S^2/\sim_{r_K}$ equipped with the quotient topology. This graph's vertices correspond to the previously mentioned points of change in the number of contour lines: every leaf corresponds to a minimum or maximum, the rest of the vertices correspond to saddles. If $r_K$ is excellent Morse then its Reeb\=/graph is a tree with inner vertices of degree 3~\cite{arnold_counting}. Figure~\ref{fig:ell_reeb} shows the Reeb\=/graph of the ellipsoid in Figure~\ref{fig:ell_equilibria}.

Due to the one\=/to\=/one correspondence from its vertices to critical points, the Reeb\=/graph of an object also determines its primary equilibrium class. On the other hand, two objects in the same primary class might not have isomorphic Reeb\=/graphs. For example see the first row of Table~\ref{tbl:twotwo} showing the different Reeb\=/graphs in the primary equilibrium class  $\{2,2\}$. An \emph{R\=/secondary equilibrium class} contains all convex bodies with isomorphic Reeb graphs.

\subsection{Morse\==Smale graph}\label{sec:ms}

By definition at any regular point~$p$ the gradient vector~$\nabla r_K$ is non\=/zero. Following the vector an \emph{integral curve} $c:\mathbb{R}\to S^2$ is traced out, which is the longest possible curve through $p$ whose derivatives agree with the gradient~$\nabla r_K$~\cite{edelsbrunner_morse_smale}.

If $r_K$ is Morse, then every regular point belongs to one and only one integral curve, two integral curves are either disjoint or exactly the same. An integral curve starts at a critical point and ends at a critical point but does not contain them. The two critical points at the ends of an integral curve are distinct because the function's value is strictly ascending along the curve. The critical point with lower critical value is called the \emph{origin} while the other one is the \emph{destination} of the integral curve. The origin and destination is either minimum\==maximum, minimum\==saddle or saddle\==maximum. Figure~\ref{fig:ell_integral} gives an example for the first one.

\begin{figure}[!ht]
  \centering
  \begin{tikzpicture}
    \fill[gray] (0, 0, -3.0) ellipse (0.0 and 0.0);
\fill[gray] (0, 0, -2.9) ellipse (0.5120763831912405 and 0.25603819159562025);
\fill[gray] (0, 0, -2.8) ellipse (0.7180219742846008 and 0.3590109871423004);
\fill[gray] (0, 0, -2.7) ellipse (0.8717797887081344 and 0.4358898943540672);
\fill[gray] (0, 0, -2.6) ellipse (0.9977753031397176 and 0.4988876515698588);
\fill[gray] (0, 0, -2.5) ellipse (1.1055415967851334 and 0.5527707983925667);
\fill[gray] (0, 0, -2.4) ellipse (1.2000000000000002 and 0.6000000000000001);
\fill[gray] (0, 0, -2.3) ellipse (1.284090685617215 and 0.6420453428086075);
\fill[gray] (0, 0, -2.2) ellipse (1.359738536958076 and 0.679869268479038);
\fill[gray] (0, 0, -2.1) ellipse (1.42828568570857 and 0.714142842854285);
\fill[gray] (0, 0, -2.0) ellipse (1.4907119849998598 and 0.7453559924999299);
\fill[gray] (0, 0, -1.9) ellipse (1.5477582354991866 and 0.7738791177495933);
\fill[gray] (0, 0, -1.7999999999999998) ellipse (1.6 and 0.8);
\fill[gray] (0, 0, -1.7) ellipse (1.6478942792411033 and 0.8239471396205517);
\fill[gray] (0, 0, -1.5999999999999999) ellipse (1.6918103387266028 and 0.8459051693633014);
\fill[gray] (0, 0, -1.5) ellipse (1.7320508075688772 and 0.8660254037844386);
\fill[gray] (0, 0, -1.4) ellipse (1.7688665548562132 and 0.8844332774281066);
\fill[gray] (0, 0, -1.2999999999999998) ellipse (1.80246744461277 and 0.901233722306385);
\fill[gray] (0, 0, -1.2) ellipse (1.8330302779823362 and 0.9165151389911681);
\fill[gray] (0, 0, -1.0999999999999999) ellipse (1.8607047649270483 and 0.9303523824635241);
\fill[gray] (0, 0, -1.0) ellipse (1.8856180831641267 and 0.9428090415820634);
\fill[gray] (0, 0, -0.8999999999999999) ellipse (1.9078784028338913 and 0.9539392014169457);
\fill[gray] (0, 0, -0.7999999999999998) ellipse (1.9275776393067947 and 0.9637888196533974);
\fill[gray] (0, 0, -0.6999999999999997) ellipse (1.9447936194419762 and 0.9723968097209881);
\fill[gray] (0, 0, -0.5999999999999996) ellipse (1.9595917942265426 and 0.9797958971132713);
\fill[gray] (0, 0, -0.5) ellipse (1.9720265943665387 and 0.9860132971832694);
\fill[gray] (0, 0, -0.3999999999999999) ellipse (1.9821424996424675 and 0.9910712498212337);
\fill[gray] (0, 0, -0.2999999999999998) ellipse (1.98997487421324 and 0.99498743710662);
\fill[gray] (0, 0, -0.19999999999999973) ellipse (1.9955506062794355 and 0.9977753031397177);
\fill[gray] (0, 0, -0.09999999999999964) ellipse (1.9988885800753267 and 0.9994442900376633);
\fill[gray] (0, 0, 0.0) ellipse (2.0 and 1.0);
\fill[gray] (0, 0, 0.10000000000000009) ellipse (1.9988885800753267 and 0.9994442900376633);
\fill[gray] (0, 0, 0.20000000000000018) ellipse (1.9955506062794353 and 0.9977753031397176);
\fill[gray] (0, 0, 0.30000000000000027) ellipse (1.98997487421324 and 0.99498743710662);
\draw[thick] (0.257941, 0.985406, 0.333272) -- (0.059918, 0.993163, 0.338475);
\draw[thick] (0.059918, 0.993163, 0.338475) -- (0.041423, 0.993344, 0.339925);
\draw[thick] (0.041423, 0.993344, 0.339925) -- (-0.047634, 0.992900, 0.349627);
\draw[thick] (-0.047634, 0.992900, 0.349627) -- (-0.158284, 0.989265, 0.368548);
\fill[gray] (0, 0, 0.40000000000000036) ellipse (1.9821424996424675 and 0.9910712498212337);
\draw[thick] (-0.158284, 0.989265, 0.368548) -- (-0.297294, 0.979651, 0.404594);
\draw[thick] (-0.297294, 0.979651, 0.404594) -- (-0.448452, 0.962292, 0.462001);
\fill[gray] (0, 0, 0.5) ellipse (1.9720265943665387 and 0.9860132971832694);
\draw[thick] (-0.448452, 0.962292, 0.462001) -- (-0.604065, 0.935792, 0.545559);
\fill[gray] (0, 0, 0.6000000000000001) ellipse (1.9595917942265426 and 0.9797958971132713);
\draw[thick] (-0.604065, 0.935792, 0.545559) -- (-0.758731, 0.898609, 0.661255);
\fill[gray] (0, 0, 0.7000000000000002) ellipse (1.9447936194419762 and 0.9723968097209881);
\fill[gray] (0, 0, 0.8000000000000003) ellipse (1.9275776393067947 and 0.9637888196533974);
\draw[thick] (-0.758731, 0.898609, 0.661255) -- (-0.905396, 0.849199, 0.815682);
\fill[gray] (0, 0, 0.9000000000000004) ellipse (1.9078784028338913 and 0.9539392014169457);
\fill[gray] (0, 0, 1.0) ellipse (1.8856180831641267 and 0.9428090415820634);
\draw[thick] (-0.905396, 0.849199, 0.815682) -- (-1.033908, 0.786302, 1.015081);
\fill[gray] (0, 0, 1.1000000000000005) ellipse (1.8607047649270483 and 0.9303523824635241);
\fill[gray] (0, 0, 1.2000000000000002) ellipse (1.833030277982336 and 0.916515138991168);
\draw[thick] (-1.033908, 0.786302, 1.015081) -- (-1.130477, 0.709234, 1.263898);
\fill[gray] (0, 0, 1.2999999999999998) ellipse (1.80246744461277 and 0.901233722306385);
\fill[gray] (0, 0, 1.4000000000000004) ellipse (1.7688665548562132 and 0.8844332774281066);
\fill[gray] (0, 0, 1.5) ellipse (1.7320508075688772 and 0.8660254037844386);
\draw[thick] (-1.130477, 0.709234, 1.263898) -- (-1.177599, 0.618169, 1.562255);
\fill[gray] (0, 0, 1.6000000000000005) ellipse (1.6918103387266024 and 0.8459051693633012);
\fill[gray] (0, 0, 1.7000000000000002) ellipse (1.6478942792411033 and 0.8239471396205517);
\fill[gray] (0, 0, 1.8000000000000007) ellipse (1.5999999999999996 and 0.7999999999999998);
\fill[gray] (0, 0, 1.9000000000000004) ellipse (1.5477582354991866 and 0.7738791177495933);
\draw[thick] (-1.177599, 0.618169, 1.562255) -- (-1.155152, 0.514538, 1.901292);
\fill[gray] (0, 0, 2.0) ellipse (1.4907119849998598 and 0.7453559924999299);
\fill[gray] (0, 0, 2.1000000000000005) ellipse (1.4282856857085697 and 0.7141428428542849);
\fill[gray] (0, 0, 2.2) ellipse (1.359738536958076 and 0.679869268479038);
\draw[thick] (-1.155152, 0.514538, 1.901292) -- (-1.045898, 0.402334, 2.254299);
\fill[gray] (0, 0, 2.3000000000000007) ellipse (1.2840906856172143 and 0.6420453428086071);
\fill[gray] (0, 0, 2.4000000000000004) ellipse (1.1999999999999997 and 0.5999999999999999);
\fill[gray] (0, 0, 2.5) ellipse (1.1055415967851334 and 0.5527707983925667);
\draw[thick] (-1.045898, 0.402334, 2.254299) -- (-0.861032, 0.294600, 2.559454);
\fill[gray] (0, 0, 2.6000000000000005) ellipse (0.9977753031397172 and 0.4988876515698586);
\fill[gray] (0, 0, 2.7) ellipse (0.8717797887081344 and 0.4358898943540672);
\draw[thick] (-0.861032, 0.294600, 2.559454) -- (-0.676699, 0.215118, 2.748308);
\fill[gray] (0, 0, 2.8000000000000007) ellipse (0.7180219742845992 and 0.3590109871422996);
\draw[thick] (-0.676699, 0.215118, 2.748308) -- (-0.532469, 0.162050, 2.850567);
\fill[gray] (0, 0, 2.9000000000000004) ellipse (0.5120763831912397 and 0.25603819159561986);
\draw[thick] (-0.532469, 0.162050, 2.850567) -- (-0.420609, 0.124385, 2.909072);
\draw[thick] (-0.420609, 0.124385, 2.909072) -- (-0.334358, 0.096894, 2.943461);
\draw[thick] (-0.334358, 0.096894, 2.943461) -- (-0.268827, 0.076739, 2.963848);
\draw[thick] (-0.268827, 0.076739, 2.963848) -- (-0.219940, 0.062052, 2.975988);
\draw[thick] (-0.219940, 0.062052, 2.975988) -- (-0.184244, 0.051484, 2.983248);
\draw[thick] (-0.184244, 0.051484, 2.983248) -- (-0.158831, 0.044018, 2.987608);
\draw[thick] (-0.158831, 0.044018, 2.987608) -- (-0.141281, 0.038862, 2.990234);
\draw[thick] (-0.141281, 0.038862, 2.990234) -- (-0.129611, 0.035399, 2.991810);
\draw[thick] (-0.129611, 0.035399, 2.991810) -- (-0.122228, 0.033146, 2.992741);
\draw[thick] (-0.122541, 0.029218, 2.993080) -- (-0.122087, 0.029264, 2.993118);
\draw[thick] (-0.122228, 0.033146, 2.992741) -- (-0.117882, 0.031732, 2.993271);
\draw[thick] (-0.122087, 0.029264, 2.993118) -- (-0.119703, 0.029489, 2.993315);
\draw[thick] (-0.119703, 0.029489, 2.993315) -- (-0.117886, 0.029649, 2.993463);
\draw[thick] (-0.117882, 0.031732, 2.993271) -- (-0.115623, 0.030879, 2.993550);
\draw[thick] (-0.117886, 0.029649, 2.993463) -- (-0.116474, 0.029780, 2.993576);
\draw[thick] (-0.116474, 0.029780, 2.993576) -- (-0.115423, 0.029913, 2.993655);
\draw[thick] (-0.115623, 0.030879, 2.993550) -- (-0.114755, 0.030380, 2.993671);
\draw[thick] (-0.114755, 0.030380, 2.993671) -- (-0.114792, 0.030091, 2.993694);
\draw[thick] (-0.115423, 0.029913, 2.993655) -- (-0.114792, 0.030091, 2.993694);
\fill[gray] (0, 0, 3.0) ellipse (0.0 and 0.0);
    \input{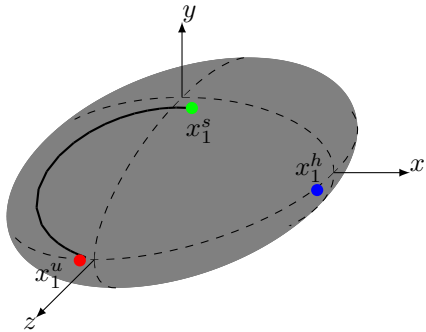}
    \begin{scope}
      \clip (-2,1) -- (2,-1) -- (2,1);
      \draw[dashed] ellipse (2 and 1);
    \end{scope}
    \begin{scope}[rotate around x=90]
      \clip (-2,3) -- (2,-3) -- (2,3);
      \draw[dashed] ellipse (2 and 3);
    \end{scope}
    \begin{scope}[rotate around y=-90]
      \clip (-3,1) -- (3,-1) -- (3,1);
      \draw[dashed] ellipse (3 and 1);
    \end{scope}
    \fill[green] (s2) circle (0.08);
    \draw ($(s2)+(0.1,-0.3)$) node {$x^s_1$};
    \fill[blue] (h2) circle (0.08);
    \draw ($(h2)+(-0.1,0.3)$) node {$x^h_1$};
    \fill[red] (u2) circle (0.08);
    \draw ($(u2)+(-0.4,-0.2)$) node {$x^u_1$};
    \draw[-latex] (2,0,0) -- (3,0,0);
    \draw (3.1,0.1,0) node {$x$};
    \draw[-latex] (0,1,0) -- (0,2,0);
    \draw (0.1,2.1,0) node {$y$};
    \draw[-latex] (0,0,3) -- (0,0,5);
    \draw (0,0,5.2) node {$z$};
  \end{tikzpicture}
  \caption{One of the integral curves from $x^s_1$ to $x^u_1$}\label{fig:ell_integral}
\end{figure}

The \emph{descending (ascending) manifold} of a maximum (minimum) critical point~$x$ is the union of $x$ and all integral curves with $x$ as their destination (origin). Only minimum\==maximum integral curves belong to both a descending and an ascending manifold, an integral curve is \emph{isolated} otherwise. Every saddle point is the origin of exactly two isolated integral curves as well as the destination of exactly two isolated integral curves. Figure~\ref{fig:ell_isolated} shows the isolated integral curves of the ellipsoid in previous examples.

\begin{figure}[!ht]
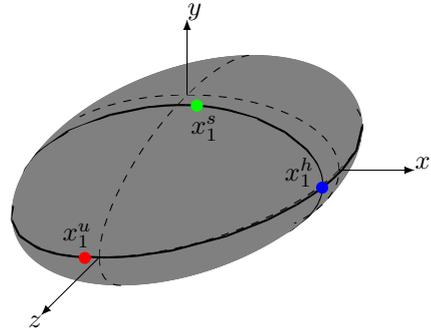

  \centering
  \begin{tikzpicture}
    \input{isolated}
    \input{equilibria}
    \begin{scope}
      \clip (-2,1) -- (2,-1) -- (2,1);
      \draw[dashed] ellipse (2 and 1);
    \end{scope}
    \begin{scope}[rotate around x=90]
      \clip (-2,3) -- (2,-3) -- (2,3);
      \draw[dashed] ellipse (2 and 3);
    \end{scope}
    \begin{scope}[rotate around y=-90]
      \clip (-3,1) -- (3,-1) -- (3,1);
      \draw[dashed] ellipse (3 and 1);
    \end{scope}
    \fill[green] (s2) circle (0.08);
    \draw ($(s2)+(0.1,-0.3)$) node {$x^s_1$};
    \fill[blue] (h2) circle (0.08);
    \draw ($(h2)+(-0.3,0.2)$) node {$x^h_1$};
    \fill[red] (u2) circle (0.08);
    \draw ($(u2)+(-0.1,0.3)$) node {$x^u_1$};
    \draw[-latex] (2,0,0) -- (3,0,0);
    \draw (3.1,0.1,0) node {$x$};
    \draw[-latex] (0,1,0) -- (0,2,0);
    \draw (0.1,2.1,0) node {$y$};
    \draw[-latex] (0,0,3) -- (0,0,5);
    \draw (0,0,5.2) node {$z$};
  \end{tikzpicture}
  \caption{Isolated integral curves}\label{fig:ell_isolated}
\end{figure}

The cells of the \emph{Morse\==Smale complex of $r_K$} are the connected components of the intersections between descending and ascending manifolds. Every cell of a Morse\==Smale complex is a quadrangle bounded by isolated integral curves connecting critical points in the following order: saddle, maximum, saddle, minimum~\cite{edelsbrunner_morse_smale}. We call the 1\=/skeleton of a Morse\==Smale complex the \emph{Morse\==Smale graph}.

The vertices in the Morse\==Smale graph of $r_K$ correspond to critical points and the edges correspond to isolated integral curves. Every vertex corresponding to a saddle is of degree 4 and it is connected to exactly two vertices corresponding to stable points and two vertices corresponding to unstable points. There are no edges in the graph other than these. The Morse\==Smale graph is a 3\=/colored quadrangular graph. You can see the Morse\==Smale graph corresponding to Figure~\ref{fig:ell_isolated} in Figure~\ref{fig:ell_ms}. The \emph{quasi\=/dual} is an alternative, equivalent representation of the Morse\==Smale graph, which is constructed from the original graph by adding stable\==unstable diagonal edges and removing all saddle\=/type vertices. The quasi\=/dual of the Morse\==Smale graph is a two\=/colored quadrangular graph~\cite{domokos_morse_smale}. This representation has fewer vertices and edges, which makes it more suitable for visualization. For this reason, we used it in the pebble catalog published as Online Resource 1.

A secondary classification method was introduced in \cite{domokos_morse_smale} based on the isomorphism classes of Morse\==Smale graphs~--~that is, shapes that belong to the same secondary class are in the same primary class as well. The authors also proved that for every combinatorially possible Morse\==Smale graph a smooth convex body exists. The Reeb\=/graph also has the first property, two shapes with isomorphic Reeb graphs belong to the same primary class. On the other hand, the Reeb\=/graph associated with a function does not uniquely define the Morse\==Smale graph associated with the same function, or vice versa.

An \emph{M\=/secondary equilibrium class} contains all convex bodies with isomorphic Morse\==Smale graphs. Objects in the same M\=/secondary class belong to the same primary class. However, bodies in the same M\=/secondary class do not necessarily belong to the same R\=/secondary class or vice versa. See Table~\ref{tbl:twotwo} for examples in the $\{2,2\}$ primary class. In the next section we introduce a tertiary classification scheme where bodies in the same tertiary class belong to the same primary, M\=/secondary and R\=/secondary class.

\subsection{Master graph}\label{sec:master}

The Morse\==Smale graph encompasses information about integral curves, the lines of fastest ascend. The Reeb\=/graph on the other hand portrays contour lines, paths of constant function value. We have introduced the so called \emph{master graph} in subsection~\ref{sss:higher} that describes both aspects. This subsection will focus on the master graph of Morse functions specifically. The key is the relationship between two types of significant curves: the isolated integral curves and the saddle contour lines.

\begin{definition}
  The \emph{intersection point} $y\in S^2$ is a point where an isolated integral curve of $\nabla r_K$ and a saddle contour of $r_K$ intersect.
\end{definition}

\begin{definition}
  Let $x_1$ and $x_2$ be critical points with an isolated integral curve~$c$ between them. Let $y_1,\dots,y_n$ be all the intersection points on $c$ ($r_K(y_1)<\dots<r_K(y_n)$). Since $r_K$ is monotonic along $c$, the ordering along $c$ will agree with the ordering by function values. An \emph{isolated set} is the ordered set of points $x_1-y_1-y_2-\dots-y_{n-1}-y_n-x_2$. Two adjacent points in this sequence are called \emph{gradient neighbors}.
\end{definition}

\begin{definition}
  A saddle point~$x^h$ and an intersection point~$y$ are \emph{contour neighbors} if there is a saddle contour line of $r_K$ that contains both $x^h$ and $y$.
\end{definition}

\begin{definition}
  In the \emph{master graph} a vertex corresponds to every critical point and every intersection point. Edges run between every gradient neighbor and every contour neighbor.
\end{definition}

We call edges between gradient neighbors \emph{gradient edges}. We call the path corresponding to an isolated set an \emph{isolated path} which is made up of gradient edges. We call edges between contour neighbors \emph{contour edges}. Figure~\ref{fig:master} shows an example.

\begin{figure}[!ht]
  \centering
  \begin{tikzpicture}
    \node[circle, draw=red, fill=white, inner sep=1pt] (U1) at (0,0) {$x^u_1$};
    \node[circle, draw=blue, fill=white, inner sep=1pt] (H1) at (1,0) {$x^h_1$};
    \node[circle, draw=red, fill=white, inner sep=1pt] (U2) at (2,0) {$x^u_2$};
    \node[circle, draw=green!70!black, fill=white, inner sep=1pt] (S1) at (3,0) {$x^s_1$};
    \node[circle, draw=blue, fill=white, inner sep=1pt] (H2) at (4,0) {$x^h_2$};
    \node[circle, draw=green!70!black, fill=white, inner sep=1pt] (S2) at (5,0) {$x^s_2$};
    
    \node[circle, fill=orange] (X1) at (4,1.2) {};
    \node[circle, fill=orange] (X2) at (4,-1.2) {};
    \node[circle, fill=orange] (X3) at (3,0.6) {};
    \node[circle, fill=orange] (X4) at (3,-0.6) {};
    
    \draw[ultra thick, red] (H1) to[bend left] (X1) (X1) to[bend left] (S2);
    \draw[dashed] (X1) to (H2);
    \draw[ultra thick, blue] (H1) to[bend right] (X2) (X2) to[bend right] (S2);
    \draw[dashed] (X2) to (H2);
    \draw[ultra thick, green] (H1) to (U1);
    \draw[ultra thick] (H1) to (U2);
    \draw[ultra thick, green] (H2) to[bend right] (X3) (X3) to[bend right] (U2);
    \draw[dashed] (X3) to[bend right] (H1);
    \draw[ultra thick, orange] (H2) to[bend left] (X4) (X4) to[bend left] (U2);
    \draw[dashed] (X4) to[bend left] (H1);
    \draw[ultra thick, purple] (H2) to (S1);
    \draw[ultra thick, yellow!70!black] (H2) to (S2);
  \end{tikzpicture}
  \caption{A master graph (\(x^u_1,x^u_2\): unstable points, \(x^h_1,x^h_2\): saddle points, \(x^s_1,x^s_2\): stable points, orange dots: intersection points, solid lines: gradient edges and solid lines of the same color constitute an isolated path, dashed lines: contour edges)}\label{fig:master}
\end{figure}
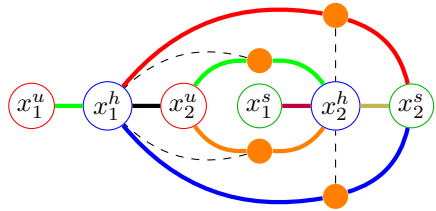

We may obtain the Morse\==Smale complex from the master graph by removing every contour edge and replacing isolated paths with single edges. We may obtain the Reeb\=/graph by contracting the contour edges and unifying parallel edges. Figure~\ref{fig:ms_and_reeb} shows the Morse\==Smale complex and Reeb\=/graph that belong to the master graph in Figure~\ref{fig:master}.

\begin{table}[!ht]
  \centering
  \caption{Secondary and tertiary classes in the $\{2,2\}$ primary class. Rows correspond to M\=/secondary classes, columns correspond to R\=/secondary classes. Intersection of row~$i$ and column~$j$ shows the master graph belonging to the $i$\=/th M\=/secondary and $j$\=/th R\=/secondary class.}\label{tbl:twotwo}
  \setlength{\tabcolsep}{0pt}
  \begin{tabular}{|c|c|c|}
    \hline
    \diagbox{MS}{Reeb}
    &\begin{tikzpicture}
      \node [circle, draw=red, fill=white, inner sep=1pt] (U1) at (0,3) {$x^u_1$};
      \node [circle, draw=red, fill=white, inner sep=1pt] (U2) at (2,3) {$x^u_2$};
      \node [circle, draw=blue, fill=white, inner sep=1pt] (H1) at (1,2) {$x^h_1$};
      \node [circle, draw=blue, fill=white, inner sep=1pt] (H2) at (1,1) {$x^h_2$};
      \node [circle, draw=green!70!black, fill=white, inner sep=1pt] (S1) at (0,0) {$x^s_1$};
      \node [circle, draw=green!70!black, fill=white, inner sep=1pt] (S2) at (2,0) {$x^s_2$};
      \draw (H1) to (U1) (H1) to (U2) (H2) to (H1) (S1) to (H2) (S2) to (H2);
    \end{tikzpicture}
    &\begin{tikzpicture}
      \node [circle, draw=red, fill=white, inner sep=1pt] (U1) at (0,3) {$x^u_2$};
      \node [circle, draw=red, fill=white, inner sep=1pt] (U2) at (0,2) {$x^u_1$};
      \node [circle, draw=blue, fill=white, inner sep=1pt] (H1) at (1,2) {$x^h_2$};
      \node [circle, draw=blue, fill=white, inner sep=1pt] (H2) at (1,1) {$x^h_1$};
      \node [circle, draw=green, fill=white, inner sep=1pt] (S1) at (2,1) {$x^s_1$};
      \node [circle, draw=green, fill=white, inner sep=1pt] (S2) at (2,0) {$x^s_2$};
      \draw (H1) to (U1) (H2) to (H1) (H2) to (U2) (S1) to (H1) (S2) to (H2);
    \end{tikzpicture}\\
    \hline
    \begin{tikzpicture}
      \node[circle, draw=red, fill=white, inner sep=1pt] (U1) at (0,0) {$x^u_1$};
      \node[circle, draw=blue, fill=white, inner sep=1pt] (H1) at (0,1) {$x^h_1$};
      \node[circle, draw=red, fill=white, inner sep=1pt] (U2) at (0,2) {$x^u_2$};
      \node[circle, draw=green!70!black, fill=white, inner sep=1pt] (S1) at (0,3) {$x^s_1$};
      \node[circle, draw=blue, fill=white, inner sep=1pt] (H2) at (0,4) {$x^h_2$};
      \node[circle, draw=green!70!black, fill=white, inner sep=1pt] (S2) at (0,5) {$x^s_2$};
      \draw (H1) to[bend right] (S1) (H1) to[bend left] (S2) (H1) to (U1) (H1) to (U2);
      \draw (H2) to (S1) (H2) to (S2) (H2) to[bend left] (U1) (H2) to[bend right] (U2);
    \end{tikzpicture}
    &\begin{tikzpicture}
      \node[circle, draw=red, fill=white, inner sep=1pt] (U1) at (0,0) {$x^u_1$};
      \node[circle, draw=blue, fill=white, inner sep=1pt] (H1) at (0,1) {$x^h_1$};
      \node[circle, draw=red, fill=white, inner sep=1pt] (U2) at (0,2) {$x^u_2$};
      \node[circle, draw=green!70!black, fill=white, inner sep=1pt] (S1) at (0,3) {$x^s_1$};
      \node[circle, draw=blue, fill=white, inner sep=1pt] (H2) at (0,4) {$x^h_2$};
      \node[circle, draw=green!70!black, fill=white, inner sep=1pt] (S2) at (0,5) {$x^s_2$};
      
      \node[circle, fill=orange] (X1) at (1.2,1) {};
      \node[circle, fill=orange] (X2) at (-1.2,4) {};
      \node[circle, fill=orange] (X3) at (0.6,2) {};
      \node[circle, fill=orange] (X4) at (-0.6,3) {};

      \draw (H1) to[bend right] (X3) (X3) to[bend right] (S1);
      \draw[dashed] (X3) to[bend right] (H2);
      \draw (H1) to[bend left] (X2) (X2) to[bend left] (S2);
      \draw[dashed] (X2) to (H2);
      \draw (H1) to (U1) (H1) to (U2);
      \draw (H2) to[bend left] (X1) (X1) to[bend left] (U1);
      \draw[dashed] (X1) to (H1);
      \draw (H2) to[bend right] (X4) (X4) to[bend right] (U2);
      \draw[dashed] (X4) to[bend right] (H1);
      \draw (H2) to (S1) (H2) to (S2);
    \end{tikzpicture}&empty\\
    \hline
    \begin{tikzpicture}
      \node[circle, draw=red, fill=white, inner sep=1pt] (U1) at (0,0) {$x^u_1$};
      \node[circle, draw=blue, fill=white, inner sep=1pt] (H1) at (0,1) {$x^h_1$};
      \node[circle, draw=red, fill=white, inner sep=1pt] (U2) at (0,2) {$x^u_2$};
      \node[circle, draw=green!70!black, fill=white, inner sep=1pt] (S1) at (0,3) {$x^s_1$};
      \node[circle, draw=blue, fill=white, inner sep=1pt] (H2) at (0,4) {$x^h_2$};
      \node[circle, draw=green!70!black, fill=white, inner sep=1pt] (S2) at (0,5) {$x^s_2$};
      \draw (H1) to[bend left] (S2) (H1) to[bend right] (S2) (H1) to (U1) (H1) to (U2);
      \draw (H2) to (S1) (H2) to (S2) (H2) to[bend right] (U2) (H2) to[bend left] (U2);
    \end{tikzpicture}
    &\begin{tikzpicture}
      \node[circle, draw=red, fill=white, inner sep=1pt] (U1) at (0,0) {$x^u_1$};
      \node[circle, draw=blue, fill=white, inner sep=1pt] (H1) at (0,1) {$x^h_1$};
      \node[circle, draw=red, fill=white, inner sep=1pt] (U2) at (0,2) {$x^u_2$};
      \node[circle, draw=green!70!black, fill=white, inner sep=1pt] (S1) at (0,3) {$x^s_1$};
      \node[circle, draw=blue, fill=white, inner sep=1pt] (H2) at (0,4) {$x^h_2$};
      \node[circle, draw=green!70!black, fill=white, inner sep=1pt] (S2) at (0,5) {$x^s_2$};
      
      \node[circle, fill=orange] (X1) at (1.2,4) {};
      \node[circle, fill=orange] (X2) at (-1.2,4) {};
      \node[circle, fill=orange] (X3) at (0.6,3) {};
      \node[circle, fill=orange] (X4) at (-0.6,3) {};

      \draw (H1) to[bend right] (X1) (X1) to[bend right] (S2);
      \draw[dashed] (X1) to (H2);
      \draw (H1) to[bend left] (X2) (X2) to[bend left] (S2);
      \draw[dashed] (X2) to (H2);
      \draw (H1) to (U1) (H1) to (U2);
      \draw (H2) to[bend left] (X3) (X3) to[bend left] (U2);
      \draw[dashed] (X3) to[bend left] (H1);
      \draw (H2) to[bend right] (X4) (X4) to[bend right] (U2);
      \draw[dashed] (X4) to[bend right] (H1);
      \draw (H2) to (S1) (H2) to (S2);
    \end{tikzpicture}
    &\begin{tikzpicture}
      \node[circle, draw=red, fill=white, inner sep=1pt] (U1) at (0,0) {$x^u_1$};
      \node[circle, draw=blue, fill=white, inner sep=1pt] (H1) at (0,1) {$x^h_1$};
      \node[circle, draw=red, fill=white, inner sep=1pt] (U2) at (0,2.5) {$x^u_2$};
      \node[circle, draw=green!70!black, fill=white, inner sep=1pt] (S1) at (0,3.5) {$x^s_1$};
      \node[circle, draw=blue, fill=white, inner sep=1pt] (H2) at (0,4.5) {$x^h_2$};
      \node[circle, draw=green!70!black, fill=white, inner sep=1pt] (S2) at (0,6) {$x^s_2$};
      
      \node[circle, fill=orange] (X1) at (0,1.75) {};
      \node[circle, fill=orange] (X2) at (0,5.25) {};

      \draw (H1) .. controls ++(1.5,1) and ++(1.5,-1) .. (S2) (H1) .. controls ++(-1.5,1) and ++(-1.5,-1) .. (S2);
      \draw (H1) to (U1) (H1) to (X1) (X1) to (U2);
      \draw[dashed] (X1) .. controls ++(0.8,1) and ++(0.8,-1) .. (H2);

      \draw (H2) to (S1) (H2) to (X2) (X2) to (S2);
      \draw[dashed] (X2) .. controls ++(-1,-1) and ++(-1,1) .. (H1);
      \draw (H2) to[bend left] (U2) (H2) to[bend right] (U2);
    \end{tikzpicture}\\
    \hline
  \end{tabular}
\end{table}

Using the master graph we can define a third level in our classification hierarchy.

\begin{definition}
  A \emph{tertiary equilibrium class} contains all convex bodies with isomorphic master graphs.
\end{definition}

Bodies in the same tertiary class belong to the same primary, M\=/secondary and R\=/secondary classes. Table~\ref{tbl:twotwo} shows how there could be two different master graphs in the same M\=/secondary and R\=/secondary classes.

\section{Classification of convex polyhedra}\label{sec:polyhedral}

In this section we examine the case where the radial distance function is not smooth, rather its image is the convex polyhedron~$P$. Just like introduced in subsection~\ref{sss:primary}, the value of the function still measures the distance from the point~$o$ which is usually chosen as the center of gravity of $P$. The function is continuous just like in the previous section, but the gradient $\nabla r_P$ exists only in the interior of the faces of $P$. If the point~$p$ is in the interior of the edge~$e$, then at $p$ only the directional derivative along $e$ exists. However, for each face adjacent to $e$ we can get the gradient of the function at $p$ that measures the distance from $o$ for the whole plane of the face. In~\cite{ludmany2021morsesmale} we named these vectors \emph{candidate gradients at~$p$} if they are tangential to the polyhedron~$P$. Candidate gradients were defined the same way at vertices too. We called points where at least one non\=/zero candidate gradient exists \emph{regular}, and also proved that at these points there is a unique candidate gradient with maximal length that we called the \emph{extended gradient $\nabla^\text{ext}r_P$}.

The concept of equilibrium points and non\=/degeneracy of convex polyhedra are already established in the literature~\cite{balancing}. We say that $x\in P$ is an \emph{equilibrium point} of $P$ (with respect to $o$) if the plane~$H$ through $x$ and perpendicular to $[o,x]$ supports $P$ at $x$. In this case $x$ is \emph{non\=/degenerate} if $H\cap P$ is the (unique) face, edge or vertex of $P$ that contains $x$ in its relative interior. We have shown in~\cite{ludmany2021morsesmale} that a point $x\in P$ is an equilibrium point if and only if there is no candidate gradient at $x$. The function~$r_P$ is \emph{polyhedral Morse} if all its equilibrium points are non\=/degenerate. The function is \emph{excellent polyhedral Morse} if the equilibrium points are at distinct distances from $o$.

Let us take a regular tetrahedron as an example. If we set $o$ in its center of gravity then it has a stable equilibrium point on all of its faces, a saddle on all of its edges and an unstable equilibrium in all of its vertices. All of these are non\=/degenerate, but all equilibria of the same kind has the same distance from $o$. Similarly to the ellipsoid previously, we can move $o$ such that the distance from it becomes an excellent polyhedral Morse function. See Figure~\ref{fig:poly_equilibria} for reference.

\begin{figure}[!ht]
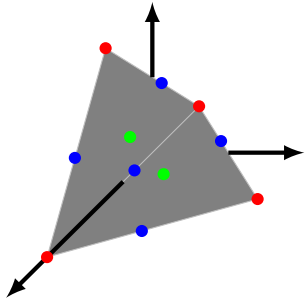

  \centering
  \begin{tikzpicture}
    \coordinate (o) at (0,0,0);
    \coordinate (v1) at (1,1,1);
    \coordinate (v2) at (-1,1,-1);
    \coordinate (v3) at (1,-1,-1);
    \coordinate (v4) at (-1,-1,1);

    \fill[gray] (v2) -- (v1) -- (v3) -- (v4) -- cycle;
    \draw[white!50!gray] (v1) -- (v2) (v1) -- (v3) (v1) -- (v4) (v3) -- (v4) (v4) -- (v2);

    \draw[ultra thick, -latex] (1,0,0) -- (2,0,0);
    \draw[ultra thick, -latex] (0,1,0) -- (0,2,0);
    \draw[ultra thick, -latex] (0,0,1) -- (0,0,5);

    \begin{scope}[fill=green]
      \input{stables}
    \end{scope}
    \begin{scope}[fill=blue]
      \input{saddles}
    \end{scope}
    \fill[red] (v1) circle (0.08);
    \fill[red] (v2) circle (0.08);
    \fill[red] (v3) circle (0.08);
    \fill[red] (v4) circle (0.08);
  \end{tikzpicture}
  \caption{Visible equilibrium points of the example tetrahedron}\label{fig:poly_equilibria}
\end{figure}

An approach of defining critical points directly on polyhedra already exists in the literature~\cite{banchoff_polyhedral} where the function is the distance from a reference plane. This function is piecewise linear and all of its critical points fall on vertices. This setup is a good fit for processing all kinds of datasets from terrains to models for 3D printing, which makes it widely used. A review of existing algorithms for computing critical points and the Morse\==Smale graph in such a case is available in~\cite{defloriani_summary}. These are not applicable to the not\=/piecewise\=/linear radial distance function~$r_P$ discussed in the current article but inspired aspects of our solution.

\subsection{Reeb\=/graph}\label{sec:poly_reeb}

Subsection~\ref{sec:reeb}'s definitions of the contour lines and the Reeb\=/graph apply to $r_P$ as well, but their properties might differ. Contour lines are continuous, closed curves, made up of circular arcs connecting at non\=/differentiable points. Contour lines are unique at regular points. Contour lines containing a stable or unstable point do not include other points. If $r_P$ is excellent polyhedral Morse then a saddle contour line contains a single saddle only. Saddle contour lines containing the non\=/degenerate saddle~$x^h$ on the edge~$e$ contain a circular arc (which might be a complete circle) on each face adjacent to $e$ that is tangent to $e$ at $x^h$. See Figure~\ref{fig:poly_saddle} for an example.

\begin{figure}[!ht]
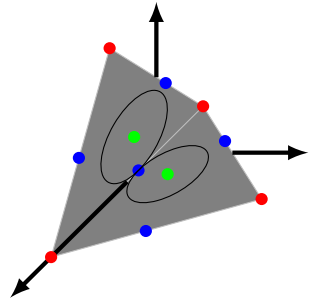

  \centering
  \begin{tikzpicture}
    \coordinate (o) at (0,0,0);
    \coordinate (v1) at (1,1,1);
    \coordinate (v2) at (-1,1,-1);
    \coordinate (v3) at (1,-1,-1);
    \coordinate (v4) at (-1,-1,1);

    \fill[gray] (v2) -- (v1) -- (v3) -- (v4) -- cycle;
    \draw[white!50!gray] (v1) -- (v2) (v1) -- (v3) (v1) -- (v4) (v3) -- (v4) (v4) -- (v2);

    \draw[ultra thick, -latex] (1,0,0) -- (2,0,0);
    \draw[ultra thick, -latex] (0,1,0) -- (0,2,0);
    \draw[ultra thick, -latex] (0,0,1) -- (0,0,5);

    \begin{scope}[fill=green]
      \input{stables}
    \end{scope}
    \begin{scope}[fill=blue]
      \input{saddles}
    \end{scope}
    \fill[red] (v1) circle (0.08);
    \fill[red] (v2) circle (0.08);
    \fill[red] (v3) circle (0.08);
    \fill[red] (v4) circle (0.08);

    \draw[rotate around y=-45.0, rotate around x=144.73561031724535] (0.282842712474619,-0.08164965809277254,-0.5773502691896257) circle (0.6123724356957945);
\draw[rotate around y=-135.0, rotate around x=144.73561031724535] (0.1414213562373095,-0.3265986323710905,0.5773502691896257) circle (0.5307227776030219);
  \end{tikzpicture}
  \caption{The saddle contour line of one of the saddle points}\label{fig:poly_saddle}
\end{figure}

The Reeb\=/graph of an excellent polyhedral Morse function $r_P$ is a tree, its vertices correspond to equilibrium points. Leaves correspond to stable and unstable, inner vertices correspond to saddle points. Inner vertices are of degree 3.

\subsection{Morse\==Smale graph}\label{sec:poly_ms}

The other key concept introduced in~\cite{ludmany2021morsesmale} beside the extended gradient was the \emph{ascending curve}. It traces out a path following the extended gradient, similarly to integral curves, but relaxes the requirement on the derivative. An ascending curve through a regular point is the longest possible curve whose right hand derivatives agree with the extended gradient $\nabla^\text{ext}r_P$. An ascending curve is a continuous open polygon. See Figure~\ref{fig:ascending} for example.

Every regular point belongs to at least one ascending curve, two ascending curves can merge, but cannot cross or split. An ascending curve starts at an equilibrium point and ends at an equilibrium point, containing the latter but not the first one. The two equilibrium points at the ends of an ascending curve are distinct because the function's value is strictly ascending along the curve. The equilibrium point with lower critical value is called the \emph{origin} while the other one is the \emph{destination} of the ascending curve. The origin and destination is either stable\==unstable, stable\==saddle or saddle\==unstable.

\begin{figure}[!ht]
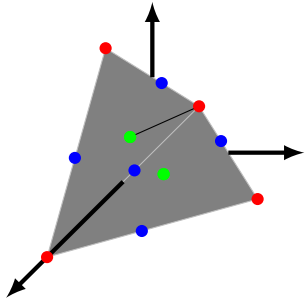

  \centering
  \begin{tikzpicture}
    \coordinate (o) at (0,0,0);
    \coordinate (v1) at (1,1,1);
    \coordinate (v2) at (-1,1,-1);
    \coordinate (v3) at (1,-1,-1);
    \coordinate (v4) at (-1,-1,1);

    \fill[gray] (v2) -- (v1) -- (v3) -- (v4) -- cycle;
    \draw[white!50!gray] (v1) -- (v2) (v1) -- (v3) (v1) -- (v4) (v3) -- (v4) (v4) -- (v2);

    \draw[ultra thick, -latex] (1,0,0) -- (2,0,0);
    \draw[ultra thick, -latex] (0,1,0) -- (0,2,0);
    \draw[ultra thick, -latex] (0,0,1) -- (0,0,5);

    \begin{scope}[fill=green]
      \input{stables}
    \end{scope}
    \begin{scope}[fill=blue]
      \input{saddles}
    \end{scope}

    \draw (s1) -- (v1);
    \begin{scope}[fill=green]
      \input{stables}
    \end{scope}

    \fill[red] (v1) circle (0.08);
    \fill[red] (v2) circle (0.08);
    \fill[red] (v3) circle (0.08);
    \fill[red] (v4) circle (0.08);
  \end{tikzpicture}
  \caption{A stable\==unstable ascending curve}\label{fig:ascending}
\end{figure}

The \emph{descending (ascending) polyhedral manifold} of an unstable (stable) equilibrium point~$x$ is the union of $x$ and all ascending curves with $x$ as their destination (origin). Only stable\==unstable ascending curves can belong to both a descending and an ascending polyhedral manifold, an ascending curve is \emph{isolated} otherwise. Every saddle point is the origin of exactly two isolated ascending curves as well as the destination of exactly two isolated ascending curves. See Figure~\ref{fig:poly_isolated} for an example.

\begin{figure}[!ht]
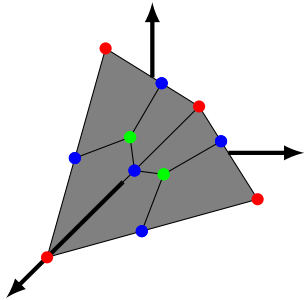

  \centering
  \begin{tikzpicture}
    \coordinate (o) at (0,0,0);
    \coordinate (v1) at (1,1,1);
    \coordinate (v2) at (-1,1,-1);
    \coordinate (v3) at (1,-1,-1);
    \coordinate (v4) at (-1,-1,1);

    \fill[gray] (v2) -- (v1) -- (v3) -- (v4) -- cycle;

    \draw[ultra thick, -latex] (1,0,0) -- (2,0,0);
    \draw[ultra thick, -latex] (0,1,0) -- (0,2,0);
    \draw[ultra thick, -latex] (0,0,1) -- (0,0,5);

    \begin{scope}[fill=green]
      \input{stables}
    \end{scope}
    \begin{scope}[fill=blue]
      \input{saddles}
    \end{scope}

    \draw (h0) -- (s1);
    \draw (h0) -- (v1);
    \draw (h0) -- (v2);
    \draw (h1) -- (s2);
    \draw (h1) -- (v1);
    \draw (h1) -- (v3);
    \draw (h2) -- (s1);
    \draw (h2) -- (s2);
    \draw (h2) -- (v1);
    \draw (h2) -- (v4);
    \draw (h4) -- (s1);
    \draw (h4) -- (v2);
    \draw (h4) -- (v4);
    \draw (h5) -- (s2);
    \draw (h5) -- (v3);
    \draw (h5) -- (v4);
    \begin{scope}[fill=green]
      \input{stables}
    \end{scope}
    \begin{scope}[fill=blue]
      \input{saddles}
    \end{scope}

    \fill[red] (v1) circle (0.08);
    \fill[red] (v2) circle (0.08);
    \fill[red] (v3) circle (0.08);
    \fill[red] (v4) circle (0.08);
  \end{tikzpicture}
  \caption{Isolated ascending curves}\label{fig:poly_isolated}
\end{figure}

The cells of the \emph{Morse\==Smale complex of $P$} are the connected components of the intersections between descending and ascending polyhedral manifolds. Every cell of the Morse\==Smale complex is a quadrangle bounded by isolated integral curves connecting equilibrium points in the following order: saddle, unstable, saddle, stable. We call the 1\=/skeleton of a Morse\==Smale complex the \emph{Morse\==Smale graph}.

The vertices in the Morse\==Smale graph of $P$ correspond to equilibrium points and the edges correspond to isolated ascending curves. Every vertex corresponding to a saddle is of degree 4 and it is connected to exactly two vertices corresponding to stable points and two vertices corresponding to unstable points. There are no edges in the graph other than these. The Morse\==Smale graph is a 3\=/colored quadrangular graph.

\subsection{Master graph}\label{sec:poly_master}

The master graph of a Morse function encompasses the information of isolated integral curves and saddle contours. In the polyhedral case discussed in this subsection it describes the relationship of isolated ascending curves and saddle contours.

\begin{definition}
  The \emph{intersection point~$y\in S^2$} is a point where an isolated ascending curve of $\nabla^\text{ext} r_P$ and a saddle contour of $r_P$ intersect.
\end{definition}

\begin{definition}
  Let $x_1$ and $x_2$ be critical points with an isolated ascending curve~$a$ between them. Let $y_1,\dots,y_n$ be all the intersection points on $a$ ($r_P(y_1)<\dots<r_P(y_n)$). An \emph{isolated set} is the ordered set of points $x_1-y_1-y_2-\dots-y_{n-1}-y_n-x_2$. Two adjacent points in this sequence are called \emph{gradient neighbors}.
\end{definition}

\begin{definition}
  A saddle point~$x^h$ and an intersection point~$y$ are \emph{contour neighbors} if there is a saddle contour line that contains both $x^h$ and $y$.
\end{definition}

Two isolated ascending curves might have merged at a point $p$ and both intersect a saddle contour at $y$ ($r_P(p)<r_P(y)$). We treat these two isolated ascending curves as separate even if they share some of their points. In the master graph we will have two separate vertices corresponding to $y$ as well.

\begin{definition}
  An intersection point~$y$ is called \emph{$n$\=/fold} if it is the intersection of $n$ isolated ascending curves and a saddle contour at the same point.
\end{definition}

\begin{definition}
  In the \emph{polyhedral master graph} a vertex corresponds to every critical point and $n$ vertices correspond to every $n$\=/fold intersection point. Every vertex is connected to one and only one of the vertices corresponding to each of its gradient neighbors. Edges run between contour neighbors too.
\end{definition}

The steps of creating the Reeb\=/graph and Morse\==Smale graph from the polyhedral master graph are identical to the smooth case.

\section{Equilibrium classes as shape catalogs}\label{sec:catalog}

So far we have defined the primary, secondary and tertiary equilibrium classes for both smooth surfaces and polyhedra. Now we can more clearly formulate the mathematical challenge outlined in Section~\ref{sec:intro}. We do this by exploring more details of the classification system in general. In subsection~\ref{sss:higher} we mentioned that $G(S,U)\leq R(S,U)\cdot M(S,U)$ for any primary class $\{S,U\}$. For example, the equality holds for trivial cases like the classes $\{1,1\}$, $\{1,2\}$ or $\{2,1\}$ with a single Reeb\=/graph, a single Morse\==Smale graph and therefore a single master graph.

Our goal in this section is to show by proving Lemma~\ref{lem:cardinality} that there is at least one primary class where the equality does not hold. We start by constructing polyhedra in 3~tertiary classes, providing a lower bound on their number. First, let us classify the polyhedron with the lowest possible number of vertices and faces. A tetrahedron with 2~stable and 2~unstable equilibria is presented in~\cite{balancing}. We use the following lemma in the classification process.

\begin{lemma}\label{lem:parallel}
  A non\==monostatic tetrahedron's radial distance function cannot have parallel isolated ascending curves in its Morse\==Smale graph.
\end{lemma}

Before proving Lemma~\ref{lem:parallel}, we provide some necessary background. In~\cite{ludmany2021morsesmale} we defined an edge of a polyhedron as \emph{followed}, if the extended gradient is parallel to the edge at every single one of its points, otherwise we called it \emph{crossed}. The names come from the behavior of an ascending curve through interior points of an edge: in the first case they follow the edge, in the second case they cross it. We have proven that this behavior is uniform at every single point of a given edge. In the following we take advantage of these two observations:
\begin{enumerate}
\item Every edge connected to an unstable point is followed.
\item On convex polyhedra, a saddle\==stable isolated ascending curve crosses only crossed edges.
\end{enumerate}

\begin{proof}
  Consider the tetrahedron with vertices $A,B,C$ and $D$.

  First we prove that two isolated ascending curves between the saddle~$x_1^h$ and the unstable point~$x_1^u$ cannot exist. Let us consider the case where the tetrahedron has one more unstable point~$x_2^u$, and assign them to vertices as $x_1^u=A$ and $x_2^u=B$. The two isolated ascending curves connecting $x_1^h$ to $x_1^u$ form a cycle on the tetrahedron. It must be 3 edges long as any shorter would not be a cycle and any longer would contain both unstable points, or in other words, the two isolated ascending curves run along the boundary of a face. This can only be the face $ACD$ as every other face contains the other unstable vertex~$B$. According to the quadrangle lemma~\cite{edelsbrunner_morse_smale}, there is a stable point~$x_1^s$ on the face $ACD$ connected to $x_1^h$ by an isolated ascending curve. The radial distance function is strictly monotone ascending from $x_1^h$ to $x_1^u$ along both isolated ascending curves running along the edges of $ACD$. However, any face containing a stable point contains at least 2~local minima when looking at the distance function along its boundary curve. This contradiction means that parallel $x_1^h-x_1^u$ isolated ascending curves cannot exist on tetrahedra with only 2~unstable points in total. The statement holds for 3 or 4~unstable points as well, because every face would contain at least 2 of them on its boundary, which makes parallel isolated ascending curves impossible.

  Next we prove that two isolated ascending curves between the saddle~$x_1^h$ and the stable point~$x_1^s$ cannot exist either. Let us consider the case where the tetrahedron has one more stable point~$x_2^s$. The parallel isolated ascending curves connecting $x_1^h$ to $x_1^s$ form a cycle and run on all 3 faces of the tetrahedron other than the one containing $x_2^s$. The same logic applies as previously: a cycle on less faces would not be possible, a cycle on more faces would have to contain $x_2^s$ too. The cycle encircles one of the vertices of the tetrahedron on its own, which must be unstable according to the quadrangle lemma. This leads to a contradiction though, because every edge connected to an unstable vertex must be \emph{followed}, but one of the $x_1^h-x_1^s$ isolated ascending curves also \emph{crossed} them. The extension to 3 or 4 stable points is also similar to the unstable case: adding more of these points would make it impossible to form a cycle of parallel isolated ascending curves on the faces of the tetrahedron.
\end{proof}

As a consequence, the previously cited tetrahedron is in the first row of Table~\ref{tbl:twotwo}. In the next step we construct polyhedra for the bottom row of the table. Start with a body that has parallel isolated ascending curves between one of its saddle and its single stable point already, the monostable polyhedron of Conway and Guy~\cite{doi:10.1137/1011014}. It has 1~stable, 3~saddle and 4~unstable points, see these in Figure~\ref{fig:conway}.

\begin{figure}[!ht]
  \centering
  \begin{tikzpicture}[scale=.65]
    \tikzmath {
      int \n, \m;
      \m=9;
      \h=1.5;
      \b=180.0/\m;
      \px=0;
      \py=\h;
      {\draw[gray!50!white] (0,0) -- (\px,\py);};
      for \n in {1,...,\m-1} {
        \r=\h*cos(\b)^\n;
        \a=90.0-\n*\b;
        \x=\r*cos(\a);
        \y=\r*sin(\a);
        {\draw (\px,\py) -- (\x,\y);};
        {\draw (-\px,\py) -- (-\x,\y);};
        {\draw[gray!50!white] (0,0) -- (\x,\y);};
        \px=\x;
        \py=\y;
      };
      {\draw (\px,\py) -- (-\px,\py);};
      {\draw[gray!50!white] (0,0) -- (\px,\py) circle (0.05);};
      {\draw[gray!50!white] (0,0) -- (0,\py) circle (0.05);};
      {\fill[blue] (0,\py) circle (0.07);};
      {\fill[red] (\px,\py) circle (0.07);};
      {\fill[red] (-\px,\py) circle (0.07);};
    }
    \fill (0,0) circle (0.05);
    \draw[gray] (-1.7,0) -- (0,0);
    \draw[gray] (-1.7,\h) -- (0,\h);
    \draw[gray,<->] (-1.5,0) -- (-1.5,\h);
    \draw[gray] (-1.7,\h/2) node {$r$};
    \begin{scope}[xshift=5cm]
      \tikzmath {
        \a=2.0;
        \s=2.0/(\h+\h*cos(\b)^\m);
        for \n in {0,...,\m} {
          \y=(\h*cos(\b)^\n)*sin(90-\n*\b);
          \z=\a+(\h-\y)*\s;
          {\draw (-\z,\y) -- (\z,\y);};
        };
        \y0=\h*sin(90);
        \z0=(\h-\y0)*\s+\a;
        \ym=(\h*cos(\b)^\m)*sin(90-\m*\b);
        \zm=(\h-\ym)*\s+\a;
        {\draw (-\z0,\y0) -- (-\zm,\ym);};
        {\draw (\z0,\y0) -- (\zm,\ym);};
      }
      \fill[blue] (0,\y0) circle (0.07);
      \fill[red] (\zm,\ym) circle (0.07);
      \fill[red] (-\zm,\ym) circle (0.07);
      \draw[gray] (0,\h) -- (0,\h+0.5);
      \draw[gray] (\a,\h) -- (\a,\h+0.5);
      \draw[gray] (\a+2,{-\h*cos(\b)^\m}) -- (\a+2,\h+0.5);
      \draw[gray,<->] (0,\h+0.3) -- (\a,\h+0.3);
      \draw[gray,<->] (\a,\h+0.3) -- (\a+2,\h+0.3);
      \draw[gray] (\a/2,\h+0.5) node {$a$};
      \draw[gray] (\a+1,\h+0.5) node {$b$};
    \end{scope}
    \begin{scope}[xshift=5cm,yshift=-3cm]
      \tikzmath {
        \a=2.0;
        \s=2.0/(\h+\h*cos(\b)^\m);
        for \n in {4,...,\m-1} {
          \x=(\h*cos(\b)^\n)*cos(90-\n*\b);
          \y=(\h*cos(\b)^\n)*sin(90-\n*\b);
          \z=\a+(\h-\y)*\s;
          {\draw (-\z,-\x) -- (\z,-\x);};
          {\draw (-\z,\x) -- (\z,\x);};
          if \n != 4 then {
            {\draw (\z,\x) -- (\pz,\px);};
            {\draw (-\pz,\px) -- (-\z,\x);};
            {\draw (\z,-\x) -- (\pz,-\px);};
            {\draw (-\pz,-\px) -- (-\z,-\x);};
          };
          \px=\x;
          \pz=\z;
        };
        {\draw (-\pz,\px) -- (-\pz,-\px);};
        {\draw (\pz,\px) -- (\pz,-\px);};
      }
      \fill[green!70!black] (0,0) circle (0.07);
      \fill[blue] (-\pz,0) circle (0.07);
      \fill[blue] (\pz,0) circle (0.07);
      \fill[red] (\pz,\px) circle (0.07);
      \fill[red] (\pz,-\px) circle (0.07);
      \fill[red] (-\pz,\px) circle (0.07);
      \fill[red] (-\pz,-\px) circle (0.07);
    \end{scope}
  \end{tikzpicture}
  \caption{Monostable polyhedron of Conway and Guy. The stable point is shown in green, saddles are shown in blue, unstable points are shown in red. The figure is not proportional, for example $r=1,a=0.1,b=30$ results in a monostable polyhedron}\label{fig:conway}
\end{figure}
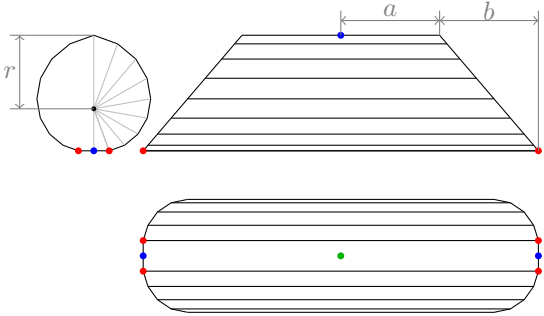

We merge the two pairs of unstable points and the saddle between them at both ends of the polyhedron. Next we raise the two resulting unstable points above their current plane by a small amount as shown in Figure~\ref{fig:split}. The resulting polyhedron is in the class $\{2,2\}$. We end up in different cells of our table depending on how much we move the edge with the saddle on the bottom. Using the notation of the previous two figures, if $c=0$, then there is two isolated ascending curves between the two saddle points, resulting in a degenerate master graph. If $c$ is greater than 0 but close to it, then the polyhedron is in the first column. If $c$ is less then $a+b$ but close to it, then the body is in the second column. If $c$ is close to neither 0 nor $a+b$ then one of the stable points vanishes, the polyhedron leaves the class $\{2,2\}$.

\begin{figure}[!ht]
  \centering
  \begin{tikzpicture}[scale=.65]
    \tikzmath {
      int \n, \m;
      \m=9;
      \h=1.5;
      \b=180.0/\m;
      \px=0;
      \py=\h;
      for \n in {1,...,\m-2} {
        \r=\h*cos(\b)^\n;
        \a=90.0-\n*\b;
        \x=\r*cos(\a);
        \y=\r*sin(\a);
        {\draw (\px,\py) -- (\x,\y);};
        {\draw (-\px,\py) -- (-\x,\y);};
        \px=\x;
        \py=\y;
      };
    }
    \draw (\px,\py) -- (0,\py-0.05) -- (-\px,\py);
    \fill[red] (0,\py-0.05) circle (0.07);  
    \begin{scope}[xshift=5cm]
      \tikzmath {
        \a=2.0;
        \s=2.0/(\h+\h*cos(\b)^\m);
        for \n in {0,...,\m - 2} {
          \y=(\h*cos(\b)^\n)*sin(90-\n*\b);
          \z=\a+(\h-\y)*\s;
          {\draw (-\z,\y) -- (\z,\y);};
        };
        \y0=\h*sin(90);
        \z0=(\h-\y0)*\s+\a;
        \y2=(\h*cos(\b)^(\m-2))*sin(90-(\m-2)*\b);
        \z2=(\h-\y2)*\s+\a;
        {\draw (-\z0,\y0) -- (-\z2,\y2);};
        {\draw (\z0,\y0) -- (\z2,\y2);};
        \y1=(\h*cos(\b)^(\m-1))*sin(90-(\m-1)*\b);
        \z1=(\h-\y1)*\s+\a;
        {\draw (-\z2,\y2) -- (1,\y1) -- (\z2,\y2);};
      }
      \fill[blue] (0,\h) circle (0.07);
      \fill[red] (-\z2,\y2) circle (0.07);  
      \fill[red] (\z2,\y2) circle (0.07);  
    \end{scope}
    \begin{scope}[xshift=5cm,yshift=-3cm]
      \draw[gray] (0,0) -- (0,1.7);
      \draw[gray] (1,0) -- (1,1.7);
      \draw[gray,<->] (0,1.5) -- (1,1.5);
      \draw[gray] (0.5,1.7) node {$c$};
      \tikzmath {
        \a=2.0;
        \s=2.0/(\h+\h*cos(\b)^\m);
        for \n in {4,...,\m-2} {
          \x=(\h*cos(\b)^\n)*cos(90-\n*\b);
          \y=(\h*cos(\b)^\n)*sin(90-\n*\b);
          \z=\a+(\h-\y)*\s;
          {\draw (-\z,-\x) -- (\z,-\x);};
          {\draw (-\z,\x) -- (\z,\x);};
          if \n != 4 then {
            {\draw (\z,\x) -- (\pz,\px);};
            {\draw (-\pz,\px) -- (-\z,\x);};
            {\draw (\z,-\x) -- (\pz,-\px);};
            {\draw (-\pz,-\px) -- (-\z,-\x);};
          };
          \px=\x;
          \pz=\z;
        };
        \n=\m-1;
        \x=(\h*cos(\b)^\n)*cos(90-\n*\b);
        \y=(\h*cos(\b)^\n)*sin(90-\n*\b);
        \z=\a+(\h-\y)*\s;
        {\draw (-\pz,\px) -- (-\z,0) -- (-\pz,-\px);};
        {\draw (\pz,\px) -- (\z,0) -- (\pz,-\px);};
        {\draw (-\pz,\px) -- (1,\x) -- (\pz,\px);};
        {\draw (-\pz,-\px) -- (1,-\x) -- (\pz,-\px);};
        {\draw (-\z,0) -- (1,\x) -- (\z,0);};
        {\draw (-\z,0) -- (1,-\x) -- (\z,0);};
        {\draw (1,\x) -- (1,-\x);};
      }
      \fill[red] (-\z,0) circle (0.07);
      \fill[red] (\z,0) circle (0.07);
      \fill[blue] (1,0) circle (0.07);
      \fill[green!70!black] (0,0) circle (0.07);
      \fill[green!70!black] (1.5,0) circle (0.07);
    \end{scope}
  \end{tikzpicture}
  \caption{Polyhedron with 2~stable and 2~unstable equilibria}\label{fig:split}
\end{figure}
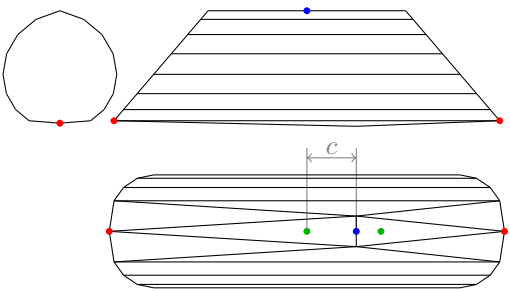

We have seen that in a non\=/degenerate Reeb\=/graph, every saddle point has at least one of its 3 neighbors at higher distance and at least one at lower distance from the center of gravity~$o$. This means that there can only be 2 non\=/isomorphic Reeb\=/graphs in $\{2,2\}$, both shown in Table~\ref{tbl:twotwo}. The same primary class also has 2 Morse\==Smale graphs~\cite{domokos_morse_smale}. However, the number of master graphs inside this class is not 4 but only 3.

\begin{proof}[Proof of Lemma~\ref{lem:cardinality}]
  We have shown that $G(2,2)\geq 3$ by providing a polyhedron in 3 of its tertiary classes.

  Let us use the notation of Table~\ref{tbl:twotwo}, and assume that there exists a master graph in the cell at the first row and second column. According to the Reeb\=/graph, the saddle contour through $x_1^h$ divides the surface into 3 parts, one containing the stable point~$x_2^s$, one containing the unstable point~$x_1^u$ and one containing the rest. Notice that the other stable point~$x_1^s$ is in this last part.

  Considering the neighbors of the saddle~$x_1^h$ in the Reeb\=/graph, only 1 of them has a lower value of $r_K$ than $r_K(x_1^h)$, namely the stable point~$x_2^s$. This means that both stable\==saddle isolated integral curves fall within the same one of the 3 previously mentioned parts bounded by the saddle contour of $x_1^h$. Because this part contains $x_2^s$ only, both curves are connected to this stable point. These two parallel edges contradict our assumption on the Morse\==Smale graph, leading to no master graph in the cell at the first row and second column of Table~\ref{tbl:twotwo}.
\end{proof}

\section{Classification of natural shapes}\label{sec:flocks}

Our last goal is the classification of natural shapes, which leads us to the solution of the algorithmic challenge mentioned in Section~\ref{sec:intro}. The equilibria of a pebble can be measured by hand, but this method is highly dependent on the person's abilities carrying out the experiment. Manual measurement of boulders gets more difficult if not outright impossible as their size increases. We utilize 3D scanning technology to achieve consistency in the results and to get around size limitations.

In such scenarios equilibria appear on two separate scales: the global value~$N$ corresponds to the approximation of the particle's convex hull by a sufficiently smooth surface while the local value~$N^\triangle$ corresponds to the polyhedral approximation (with faces of maximal diameter~$\triangle$) acquired by 3D scanning. The main obstacle is that in general $\lim_{\triangle\to 0}N^\triangle>N$~\cite{domokos_discretized, domokos_rocking}.

Local equilibria appear in \emph{flocks} which are spatially localized around the locations of global equilibria. This poses a considerable problem for measurements: while only $N^\triangle$ is directly available from any 3D scanned dataset, the physically relevant quantity is $N$. To obtain $N$ based on $N^\triangle$, one needs some artificial ``blurring'' of the data.

Edelsbrunner et.al.~achieved this by defining a simplifying operation on the Morse\==Smale graph called \emph{cancellation}, which eliminates two adjacent critical points~\cite{edelsbrunner_morse_smale}. First we describe the original operation, then present our extension of it to the master graph. There are two possible combinations of critical points to be cancelled: a minimum and a saddle or a saddle and a maximum, but they are symmetrical. Let $x^h$ be the saddle and $x^u_1$ the unstable vertex of the canceled pair while $x^u_2$ the other unstable vertex connected to $x^h$ ($x^u_1\neq x^u_2$). The cancellation combines the three vertices into $x^u_2$ by removing every edge connected to $x^h$ and merging the $x^u_1$ and $x^u_2$ vertices.

The basic idea is the same on the master graph, but only the two cancelled vertices must be connected directly, the remaining stable or unstable point can be connected to the cancelled saddle through an isolated path. This requires some extra bookkeeping in the following definition to end up with a valid master graph after cancellation.

\begin{definition}\label{def:cancel}
  Let $x^h$ be a saddle and $x^u_1$ an unstable point such that there is an $x^h-x^u_1$ edge in the master graph. Let $x^u_2$ be the other unstable point that is connected to $x^h$ through an isolated path ($x^u_1\neq x^u_2$), let the intersection points on this path be $y_1,\dots,y_n$. The \emph{cancellation of~$x^h$} merges $x^u_1$ and $x^h$ into $x^u_2$. The critical value in the remaining vertex is $r_K(x^u_2):=\max\{r_K(x^u_1),r_K(x^u_2)\}$. The steps are the following:
  \begin{enumerate}
  \item Remove every intersection along the two isolated paths originating from $x^h$ and ending in a stable vertex.
  \item Let $D$ be the set of the third\=/to\=/last vertices on every isolated path ending in $x^u_1$. Every isolated path ending in $x^u_1$~--~except for the $x^h-x^u_1$ edge~--~is at least 3~vertices long because they originate from a saddle, end in $x^u_1$ and intersect the contour line through $x^h$.
  \item Remove $x^u_1$ and all of its neighbors.
  \item Copy the $y_1,\dots,y_n$ vertices and the edges connected to them $|D|-1$ times and connect a unique copy of $y_1$ to every vertex in $D$. These new edges are not contour edges.
  \end{enumerate}
\end{definition}

Figure~\ref{fig:cancellation} shows an example.

\begin{figure}[!ht]
  \begin{subfigure}[b]{\columnwidth}
    \centering
    \begin{tikzpicture}[dot/.style={circle, draw=black, inner sep=0, minimum size=.5cm}]
      \node[dot] (Xa1) at (-1,1) {$y'_1$};
      \node[dot] (Xa2) at (-2,1) {$y'_2$};
      \node[dot] (Xs1) at (-1,-1) {$y^*_1$};
      \node[dot] (Xs2) at (-2,-1) {$y^*_2$};
      \node[dot] (Ul) at (0,0) {$x^u_1$};
      \node[dot] (H) at (1,0) {$x^h$};
      \node[dot] (X1) at (2,0) {$y_1$};
      \node[dot] (Uh) at (3,0) {$x^u_2$};
      \node[dot] (S1) at (1,1.5) {$x^s_1$};
      \node[dot] (S2) at (1,-1.5) {$x^s_2$};
      \node[dot] (H1) at (4,1) {$x^h_1$};
      
      \draw (H) to (Ul);
      \draw (H) to (X1);
      \draw (X1) to (Uh);
      \draw (Xa1) to (Ul);
      \draw (Xs1) to (Ul);
      \draw (H1) to (Uh);
      \draw (Xa1) to (Xa2);
      \draw (Xs1) to (Xs2);
      \draw[dashed] (Xa1) to (H);
      \draw[dashed] (Xs1) to (H);
      \draw[dashed] (X1) to (H1);
      \draw[decorate, decoration=snake] (H) to (S1);
      \draw[decorate, decoration=snake] (H) to (S2);
    \end{tikzpicture}
    \caption{Part of a master graph. Dashed: contour edge, squiggly: isolated path with internal vertices omitted}
  \end{subfigure}\\
  \begin{subfigure}[b]{\columnwidth}
    \centering
    \begin{tikzpicture}[dot/.style={circle, draw=black, inner sep=0, minimum size=.5cm}]
      \node[dot] (Xa1) at (-1,1) {$y'_1$};
      \node[dot] (Xa2) at (-2,1) {$y'_2$};
      \node[dot] (Xs1) at (-1,-1) {$y^*_1$};
      \node[dot] (Xs2) at (-2,-1) {$y^*_2$};
      \node[dot] (Ul) at (0,0) {$x^u_1$};
      \node[dot] (H) at (1,0) {$x^h$};
      \node[dot] (X1) at (2,0) {$y_1$};
      \node[dot] (Uh) at (3,0) {$x^u_2$};
      \node[dot] (S1) at (1,1.5) {$x^s_1$};
      \node[dot] (S2) at (1,-1.5) {$x^s_2$};
      \node[dot] (H1) at (4,1) {$x^h_1$};
      
      \draw (H) to (Ul);
      \draw (H) to (X1);
      \draw (X1) to (Uh);
      \draw (Xa1) to (Ul);
      \draw (Xs1) to (Ul);
      \draw (H1) to (Uh);
      \draw (Xa1) to (Xa2);
      \draw (Xs1) to (Xs2);
      \draw[dashed] (Xa1) to (H);
      \draw[dashed] (Xs1) to (H);
      \draw[dashed] (X1) to (H1);
    \end{tikzpicture}
    \caption{The graph after step 1. \(D=\{y'_2, y^*_2\}\)}
  \end{subfigure}\\
  \begin{subfigure}[b]{\columnwidth}
    \centering
    \begin{tikzpicture}[dot/.style={circle, draw=black, inner sep=0, minimum size=.5cm}]
      \node[dot] (Xa2) at (-2,1) {$y'_2$};
      \node[dot] (Xs2) at (-2,-1) {$y^*_2$};
      \node[dot] (X1) at (2,0) {$y_1$};
      \node[dot] (Uh) at (3,0) {$x^u_2$};
      \node[dot] (S1) at (1,1.5) {$x^s_1$};
      \node[dot] (S2) at (1,-1.5) {$x^s_2$};
      \node[dot] (H1) at (4,1) {$x^h_1$};
      
      \draw (X1) to (Uh);
      \draw (H1) to (Uh);
      \draw[dashed] (X1) to (H1);
    \end{tikzpicture}
    \caption{The graph after step 3.}
  \end{subfigure}\\
  \begin{subfigure}[b]{\columnwidth}
    \centering
    \begin{tikzpicture}[dot/.style={circle, draw=black, inner sep=0, minimum size=.5cm}]
      \node[dot] (Xa1) at (2,1) {$\ddot{y}_1$};
      \node[dot] (Xa2) at (-2,1) {$y'_2$};
      \node[dot] (Xs1) at (2,-1) {$\hat{y}_1$};
      \node[dot] (Xs2) at (-2,-1) {$y^*_2$};
      \node[dot] (Uh) at (3,0) {$x^u_2$};
      \node[dot] (S1) at (1,1.5) {$x^s_1$};
      \node[dot] (S2) at (1,-1.5) {$x^s_2$};
      \node[dot] (H1) at (4,1) {$x^h_1$};
      
      \draw (H1) to (Uh);
      \draw (Xa1) to (Xa2);
      \draw (Xs1) to (Xs2);
      \draw (Xa1) to (Uh);
      \draw (Xs1) to (Uh);
      \draw[dashed] (Xa1) to (H1);
      \draw[dashed] (Xs1) to [bend right] (H1);
    \end{tikzpicture}
    \caption{The resulting graph}
  \end{subfigure}
  \caption{Canceling the saddle~$x^h$}\label{fig:cancellation}
\end{figure}
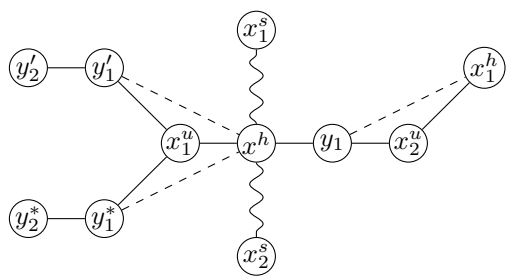
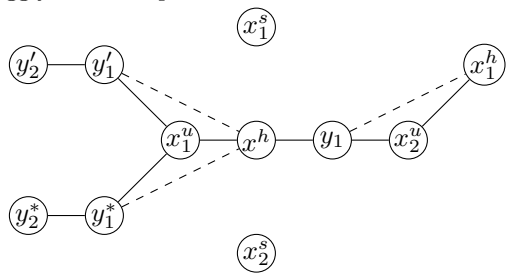
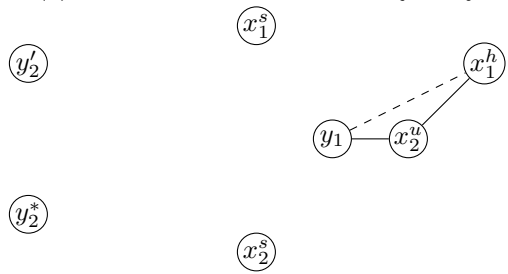
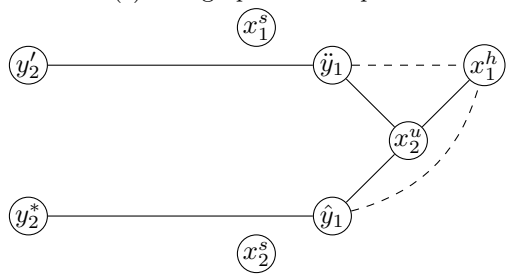

Now that we defined the basic simplification step, the next question becomes that which of the equilibria we should cancel and in what order. We draw our inspiration from the same article as before, where the authors utilized the level set of the surface. They called a critical point \emph{positive} if it created a new contour line and \emph{negative} if it destroyed one while ascending in the codomain of the distance function. Minima are positive, maxima are negative and saddles can be either positive or negative. Every negative saddle was paired with the preceding positive minimum and every negative maximum was paired with the preceding positive saddle. They defined the \emph{persistence} of a pair of equilibria as the difference in function value between the two points. Saddles were then canceled from lower to higher persistence until the desired Morse\==Smale complex was reached.

These steps are easily adapted to the master graph, as it contains the necessary information from the level sets of the surface. At any given point, the saddle with the lowest persistence is connected directly to a stable or an unstable point. There is always at least one such edge in the graph. We choose the next equilibria to be canceled by iterating over all the remaining edges of the graph looking for the saddle\==stable or saddle\==unstable edge having the lowest difference in the value of $r_K$ at its vertices. We cancel the saddles until we get the closest possible to the number of stable and unstable points measured by hand.

\subsection{Example on a single pebble}
% Ez a K1_2.stl alapján készült.

In this section we show the steps the algorithm takes to classify a single pebble. First it was established via manual measurements that this pebble has 2~stable and 3~unstable equilibrium points, placing it in the $\{2, 3\}$ primary class. Next the pebble was scanned using a 3D scanner, resulting in the polyhedron denoted from now on by $P_\text{ex}$. Determining the higher order equilibrium classes consists of three main steps:
\begin{enumerate*}[label=\arabic*)]
\item constructing the master graph of $P_\text{ex}$,
\item sorting the saddles and canceling them up to the point where the master graph has only 2~stable and 3~unstable points,
\item creating the Reeb\=/graph and Morse\==Smale graph from the master graph.
\end{enumerate*}
In the rest of this section we will discuss these steps in detail. The source code of the algorithm's implementation is also available~\cite{balazs_ludmany_2023_7599635}.

\subsubsection[Construction of the master graph of $P_\text{ex}$]{Construction of the master graph of $\boldsymbol{P_\text{ex}}$}

In the very first step the algorithm goes over all the edges of $P_\text{ex}$ and determines if they contain a saddle point. There are 25 such edges with saddle points on them in this example. This results in a list of $(e_i,x^h_i)$ pairs where $e_i$ is the edge containing the saddle point~$x^h_i$ and $i$ goes from 1 to 25.

The next step is tracing the saddle contour line containing the saddle~$x^h_i$ for every $i\in [1, 25]$. As established in Section~\ref{sec:poly_reeb}, contour lines are closed, continuous curves made up of circular arcs. Every arc is either a full circle on a single face or has its endpoints on edges of $P_\text{ex}$. We can therefore trace a contour line in an iterative fashion by always finding the next arc connected to one of the endpoints of a partial contour line until the curve is closed.

We have also established previously that the saddle contour line containing any $x^h_i$ has an arc on each face adjacent to $e_i$. Any of these two can be the starting partial contour line in the previous paragraph's iterative method. The resulting saddle contour lines are shown in blue in Figure~\ref{fig:poly_master}. You can also take a closer look at one of the flocks in Figure~\ref{fig:poly_master_cropped}.

Once we have all the saddle contours, the next step is tracing the isolated ascending curves. Every saddle point is the origin of 2 such curves and the destination of 2 such curves. We start from the point~$x^h_i$ for every $i\in[1,25]$ and trace all 4 curves with an iterative algorithm, one segment at a time. One curve is constructed toward each endpoint of $e_i$ in the direction of the extended gradient, and one curve is constructed toward each face adjacent to $e_i$ in the direction opposite to the extended gradient. For both types of curves the last iteration stops at an unstable or stable point, respectively. The two figures mentioned in the previous paragraph show these curves as well.

\begin{figure}[!ht]
  \centering
  \begin{tikzpicture}
  \draw (0,0) node {\includegraphics{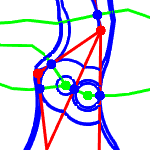}};
  \draw (0.9,-0.7) -- (3,-0.2) ++(0.3,0) node {\(x^h_1\)};
  \draw (0,-0.5) -- (3,0.2) ++(0.3,0) node {\(x^h_2\)};
  \draw (-0.8,0.4) -- (3,0.6) ++(0.3,0) node {\(x^h_3\)};
  \draw (0.9,1.6) -- (3,1) ++(0.3,0) node {\(x^u\)};
  \draw (0.8,2.1) -- (3,1.4) ++(0.3,0) node {\(x^h_4\)};
  \end{tikzpicture}
  \caption{A flock of equilibria. Stable, unstable and saddle points are represented by green, red and blue dots respectively. Saddle contour lines are shown in blue, green lines are stable\==saddle, red lines are saddle\==unstable ascending curves}\label{fig:poly_master_cropped}
\end{figure}
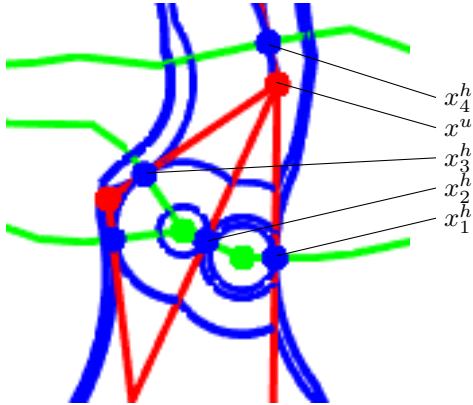

The creation of the master graph starts by adding the vertices~$x^h_i$ for every $i\in[1,25]$. We refer to the vertices of the master graph by the same name we used for their corresponding points on $P_\text{ex}$. The graph is extended every time an isolated ascending curve is traced, take for example the one between $x^h_1$ and $x^u$ in Figure~\ref{fig:poly_master_cropped}. The $x^h_1-x^u$ curve intersects the saddle contours of $x^h_2, x^h_3$ and $x^h_4$, therefore we add intersection points denoted by $y_2, y_3$ and $y_4$ to the graph. Every intersection point is then connected to the saddle with the same index by a contour edge. If this is the first time we encounter $x^u$ we assign a new vertex to it in the master graph. The last step is adding gradient edges along the $x^h_1-y_2-y_3-y_4-x^u$ path.

\subsubsection{Cancellation of saddles}

At this point we have the master graph of the polyhedron $P_\text{ex}$ with 14 stable and 13 unstable vertices. We sort the saddles and apply cancellations as described in Section~\ref{sec:flocks} until we get as close to the hand measurements as possible. In this particular case we can get a graph exactly in the $\{2,3\}$ primary class. Figure~\ref{fig:example_master} shows the master graph after the cancellations.

It is important to mention that cancellation works on the master graph, it does not modify the polyhedron itself. After any number of cancellations, the contour lines passing through the remaining saddles can still be visualized on the polyhedron. On the other hand, there is no one\=/to\=/one correspondence from edges of the master graph to isolated ascending curves on the surface of $P_\text{ex}$ anymore. The method we use to circumvent this issue is as follows (using the notations of definition~\ref{def:cancel}): every isolated ascending curve with its destination at the unstable point~$x^u_1$ is extended with the segments of both the $x^h-x^u_1$ and the $x^h-x^u_2$ isolated ascending curves before cancellation. The resulting curve is not strictly ascending anymore, but in our experience it gives a good intuition of where the corresponding isolated integral curve on the original body would be. Figure~\ref{fig:cancelled_master} shows this kind of visualization of the master graph of $P_\text{ex}$ after the last cancellation.

\subsubsection{Creation of the Reeb\=/graph and Morse\==Smale graph}

The steps to get the Reeb\=/graph and the Morse\==Smale graph from the master graph have already been described in Section~\ref{sec:master}. The resulting graphs for $P_\text{ex}$ can be seen in Figure~\ref{fig:example_graphs}.

\begin{figure}[!ht]
  \centering
  \begin{subfigure}[t]{.3\textwidth}
    \centering
    \begin{tikzpicture}[scale=.5]
        \node[circle, fill=green, minimum size=.2cm] (S1) at (1,-1) {};
        \node[circle, fill=green, minimum size=.2cm] (S2) at (5,-1) {};
    
        \node[circle, fill=blue, minimum size=.2cm] (H1) at (3,0.5) {};
        \node[circle, fill=orange, minimum size=.2cm] (I1) at (0,0.5) {};
        \node[circle, fill=orange, minimum size=.2cm] (I2) at (1,0.5) {};
        \node[circle, fill=orange, minimum size=.2cm] (I3) at (2,0.5) {};
        \node[circle, fill=orange, minimum size=.2cm] (I4) at (4,0.5) {};
        \node[circle, fill=orange, minimum size=.2cm] (I5) at (5,0.5) {};
        \node[circle, fill=orange, minimum size=.2cm] (I6) at (6,0.5) {};
    
        \node[circle, fill=blue, minimum size=.2cm] (H2) at (3,2) {};
        \node[circle, fill=orange, minimum size=.2cm] (I7) at (0,2) {};
        \node[circle, fill=orange, minimum size=.2cm] (I8) at (1,2) {};
        \node[circle, fill=orange, minimum size=.2cm] (I9) at (2,2) {};
        \node[circle, fill=orange, minimum size=.2cm] (I10) at (4,2) {};
        \node[circle, fill=orange, minimum size=.2cm] (I11) at (5,2) {};
        \node[circle, fill=orange, minimum size=.2cm] (I12) at (6,2) {};
    
        \node[circle, fill=blue, minimum size=.2cm] (H3) at (1,3) {};
        \node[circle, fill=orange, minimum size=.2cm] (I13) at (0,3) {};
        \node[circle, fill=orange, minimum size=.2cm] (I14) at (2,3) {};
    
        \node[circle, fill=blue, minimum size=.2cm] (H4) at (5,3) {};
        \node[circle, fill=orange, minimum size=.2cm] (I15) at (4,3) {};
        \node[circle, fill=orange, minimum size=.2cm] (I16) at (6,3) {};
    
        \node[circle, fill=red, minimum size=.2cm] (U1) at (0,4) {};
        \node[circle, fill=red, minimum size=.2cm] (U2) at (2,4) {};
        \node[circle, fill=red, minimum size=.2cm] (U3) at (4,4) {};
        \node[circle, fill=red, minimum size=.2cm] (U4) at (6,4) {};
        
        %H1
        \draw[very thick, blue] (H1) to[bend left] (I1) (H1) to[bend right] (I2) (H1) to (I3) (H1) to (I4) (H1) to[bend left] (I5) (H1) to[bend right] (I6);
        \draw[very thick, red] (H1) to (I7) to (I13) to (U1) (H1) to (I12) to (I16) to (U4);
        \draw[very thick, green] (H1) to (S1) (H1) to (S2);
    
        %H2
        \draw[very thick, blue] (H2) to[bend left] (I7) (H2) to[bend right] (I8) (H2) to (I9) (H2) to (I10) (H2) to[bend left] (I11) (H2) to[bend right] (I12);
        \draw[very thick, red] (H2) to (I14) to (U2) (H2) to (I15) to (U3);
        \draw[very thick, green] (H2) to (I3) to (S1) (H2) to (I4) to (S2);
    
        %H3
        \draw[very thick, blue] (H3) to (I13) (H3) to (I14);
        \draw[very thick, red] (H3) to (U1) (H3) to (U2);
        \draw[very thick, green] (H3) to (I8) to (I1) to (S1) (H3) to (I9) to (I2) to (S2);
    
        %H4
        \draw[very thick, blue] (H4) to (I15) (H4) to (I16);
        \draw[very thick, red] (H4) to (U3) (H4) to (U4);
        \draw[very thick, green] (H4) to (I10) to (I5) to (S1) (H4) to (I11) to (I6) to (S2);
    \end{tikzpicture}
    \caption{Master graph. Contour edges are in blue, green edges are on a stable\==saddle isolated path, red edges are on a saddle\==unstable isolated path}\label{fig:example_master}
  \end{subfigure}
  \begin{subfigure}[t]{.3\textwidth}
    \centering
    \begin{tikzpicture}[scale=.5]
        \node[circle, fill=green, minimum size=.2cm] (S1) at (1,-1) {};
        \node[circle, fill=green, minimum size=.2cm] (S2) at (5,-1) {};

        \node[circle, fill=blue, minimum size=.2cm] (H1) at (3,0.5) {};
        \node[circle, fill=blue, minimum size=.2cm] (H2) at (3,2) {};
        \node[circle, fill=blue, minimum size=.2cm] (H3) at (1,3) {};
        \node[circle, fill=blue, minimum size=.2cm] (H4) at (5,3) {};

        \node[circle, fill=red, minimum size=.2cm] (U1) at (0,4) {};
        \node[circle, fill=red, minimum size=.2cm] (U2) at (2,4) {};
        \node[circle, fill=red, minimum size=.2cm] (U3) at (4,4) {};
        \node[circle, fill=red, minimum size=.2cm] (U4) at (6,4) {};

        \draw[very thick] (S1) to (H1) (S2) to (H1) (H1) to (H2) (H2) to (H3) (H2) to (H4) (H3) to (U1) (H3) to (U2) (H4) to (U3) (H4) to (U4);
    \end{tikzpicture}
    \caption{Reeb\=/graph}
  \end{subfigure}
  \begin{subfigure}[t]{.3\textwidth}
    \centering
    \begin{tikzpicture}[scale=.5]
        \node[circle, fill=green, minimum size=.2cm] (S1) at (1,-1) {};
        \node[circle, fill=green, minimum size=.2cm] (S2) at (5,-1) {};
    
        \node[circle, fill=blue, minimum size=.2cm] (H1) at (3,0.5) {};
        \node[circle, fill=blue, minimum size=.2cm] (H2) at (3,2) {};
        \node[circle, fill=blue, minimum size=.2cm] (H3) at (1,3) {};
        \node[circle, fill=blue, minimum size=.2cm] (H4) at (5,3) {};
    
        \node[circle, fill=red, minimum size=.2cm] (U1) at (0,4) {};
        \node[circle, fill=red, minimum size=.2cm] (U2) at (2,4) {};
        \node[circle, fill=red, minimum size=.2cm] (U3) at (4,4) {};
        \node[circle, fill=red, minimum size=.2cm] (U4) at (6,4) {};
    
        \draw[very thick, green] (H1) to (S1) (H1) to (S2);
        \draw[very thick, red] (H1) to[bend left] (U1) (H1) to[bend right] (U4);
        \draw[very thick, green] (H2) to (S1) (H2) to (S2);
        \draw[very thick, red] (H2) to (U2) (H2) to (U3);
        \draw[very thick, green] (H3) to (S1) (H3) to[bend right] (S2);
        \draw[very thick, red] (H3) to (U1) (H3) to (U2);
        \draw[very thick, green] (H4) to[bend left] (S1) (H4) to (S2);
        \draw[very thick, red] (H4) to (U3) (H4) to (U4);
    \end{tikzpicture}
    \caption{Morse\==Smale graph. Green edges represent stable\==saddle isolated integral curves, red edges represent saddle\==unstable isolated integral curves}
  \end{subfigure}
  \caption{Graphs of $P_\text{ex}$. Vertices in the same position in all 3~subfigures correspond to the same equilibrium point. Stable, saddle, unstable and intersection points are represented by green, blue, red and orange points, respectively}\label{fig:example_graphs}
\end{figure}
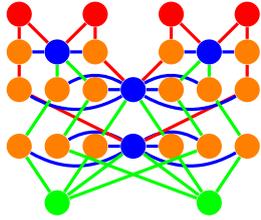
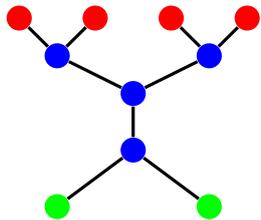
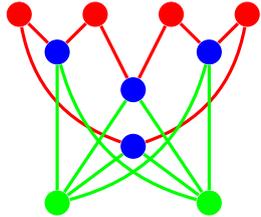

\subsection{Cataloging a population of pebbles}\label{sec:bias}

In the previous example we purposefully selected a pebble where~--~after several cancellations~--~the number of computed stable and unstable equilibria perfectly matched the hand measurements. This is not always the case. Take for instance a pebble that was placed in the class $\{3,4\}$, with its master graph having 4 stable and 4 unstable vertices at an intermediate step of the algorithm. The number of stable points is off by 1, and the next cancellation can either decrease this or decrease the number of unstable points. In the first case we get a perfect match, but in the second case both values would be off by 1, resulting in a worse match.

We have to keep in mind that despite all their best efforts, the people measuring the pebbles by hand can make mistakes. Consequently a perfect match is not inherently possible or even desirable in every single case. Following the previous paragraph's logic, we computed the master graph that is off by the smallest amount for 271 pebbles in our inventory. More formally, after every cancellation we calculated the absolute difference of the manual and the computer measurements in both the number of stable and the number of unstable points, and selected the master graph where the sum of these two absolute differences was the lowest.

A full list of every single pebble with a photo and a host of manually measured and computed values of both classical and mechanical shape descriptors is available in Online Resource 1. The 3D scan of every pebble was also published in a repository~\cite{balazs_ludmany_2023_7609228} with the \LaTeX~source code of the aforementioned document and the scripts compiling it. Be aware that some of the pebbles could not be photographed or their material could not be identified, in this case the corresponding field is left empty. The visual representation of graphs in the document are for illustration purposes only, our program outputs an alphanumerical encoding of every graph which can be used for actual classification.

Table~\ref{tbl:stat} gives the number of pebbles with a given difference in the number of computed stable/unstable points compared to the hand measurements. In 185 out of 271 cases ($68.27\%$) the computed master graph perfectly matched the measured primary class of the corresponding pebble. There are no pebbles where both values are positive because further cancellations always get results closer to the hand measurements. There are also no cases where both values are negative because further cancellations would only increase the sum of their absolute values.

\begin{table}[!ht]
\centering
\caption{The number of pebbles with a given difference in the number of computed stable/unstable points compared to the hand measurements}\label{tbl:stat}
\begin{tabular}{p{1.5cm}|*{5}{r}|r}
\multirow{2}{\linewidth}{stable difference}&\multicolumn{5}{c|}{unstable difference}&\multirow{2}{*}{sum}\\
&-1&0&1&2&3\\
\hline
\hfill -3&0&0&0&0&0&0\\
\hfill -2&0&0&1&1&0&2\\
\hfill -1&0&27&4&4&0&35\\
\hfill 0&4&185&32&10&3&234\\
\hfill 1&0&0&0&0&0&0\\
\hline
sum&4&212&37&15&3&271
\end{tabular}
\end{table}

The most important point is that 248 out of 271 pebbles ($92.51\%$) is off by at most 1 stable or 1 unstable point compared to the manual measurements. With the primary classification being this close we are confident in the computed secondary and tertiary classes as well.

In all 52 cases where the manual measurement put a pebble in the primary class $\{2,2\}$, the computer did so as well. Our most important observation is that all of them belong to the same tertiary~--~and thus R\=/secondary and M\=/secondary~--~equilibrium class, namely the one containing all the tetrahedra inside of this primary class.

Based on the mathematical results presented in Section~\ref{sec:catalog} and the measurement data in this subsection, the higher order mechanical descriptors give a finite, complete, biased natural catalog inside the primary class \(\{2,2\}\). Two of these adjectives outline the direction of further research:
\begin{enumerate*}[label=\alph*)]
    \item from the mathematical point, the completeness of the tertiary classification is still an open question,
    \item from the geomorphological point, the bias of the higher order classification can be further examined.
\end{enumerate*}

\section{Summary}
In this paper we showed that higher order mechanical descriptors are a viable option for the description of sedimentary particles. These descriptors, given as graphs, carry essential, naturally encoded three\=/dimensional information on the shape. Despite the fact that their mathematical existence was known, the challenges connected with their extraction and identification proved to be, until now, prohibitive.

First order mechanical descriptors, defined by the respective numbers of stable and unstable static balance points
of the scalar, radial distance function $r=r(\varphi, \theta)$ measured from the center of mass~\(o\), have already proven their utility in geomorphology. Second order descriptors carry deeper information on the \emph{relative position} of these points by using natural, discrete decompositions of the radial distance function $r=r(\varphi, \theta)$. The decomposition of its range by saddle points leads to the concept of Reeb\=/graphs while the decompositon of its domain by isolated integral curves of the gradient leads to the concept of Morse\==Smale graphs. While both second-order descriptors have been discussed before in the mathematical literature, still, their mutual relationship remained unclear.

Since 3D measurement technology is becoming increasingly accessible and standard, obtaining 3D datasets is not a problem any more. However, 3D data itself, if not coupled with geometric ideas, does not solve the basic question of \emph{how to describe shapes}. Encouraged by these technological developments, in this paper we presented the geometric background of the mentioned second-order descriptors, provided an algorithm to reliably extract them from scanned 3D point clouds.  By introducing a \emph{third-order descriptor}, called the
master graph, we established the relationship between Reeb\=/graphs and Morse\==Smale graphs.

To illustrate the feasibility of the application, we created a catalog  of 271 scanned pebbles where we performed these measurements and we also provided the source code for the implementation of our algorithm as well as the 3D datasets of the aforementioned pebbles.

We hope that our paper gives a signal to geomorphologists that new, geometrically inspired tools for fully three\=/dimensional shape analysis are now available and ready to deploy.

\appendix
\onecolumn
\section{Visualization of the algorithm's output}
\begin{figure}[!ht]
    \centering
    \begin{tikzpicture}
        \node at (0,8) {\includegraphics[width=8cm]{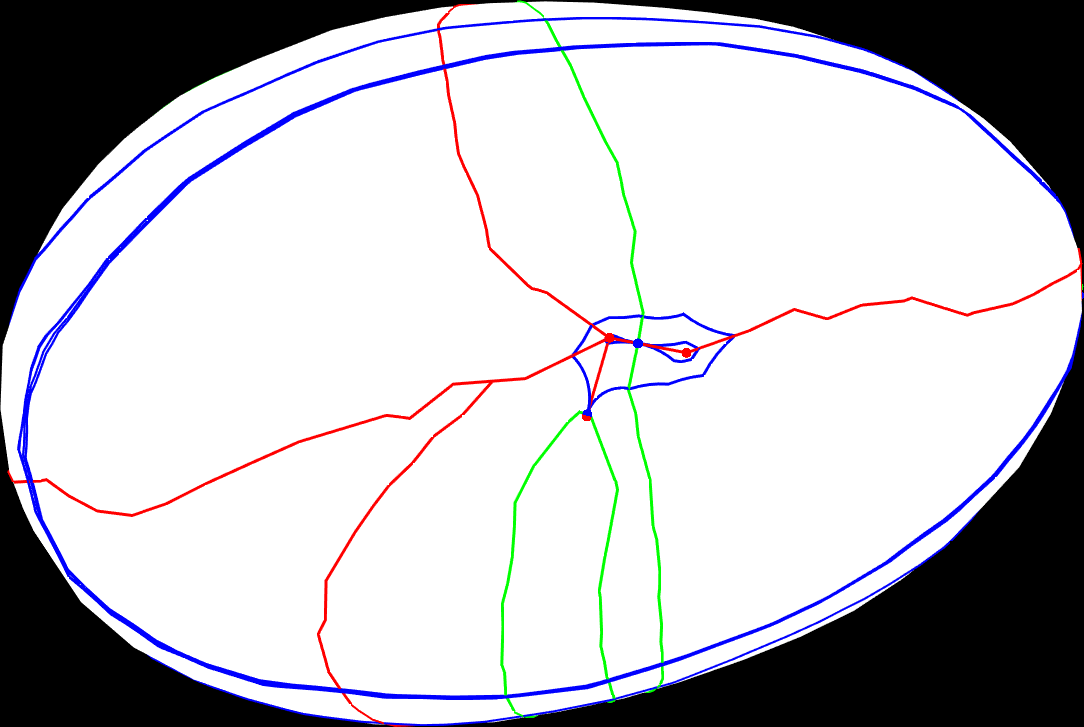}};
        \node at (0,0) {\includegraphics[width=8cm]{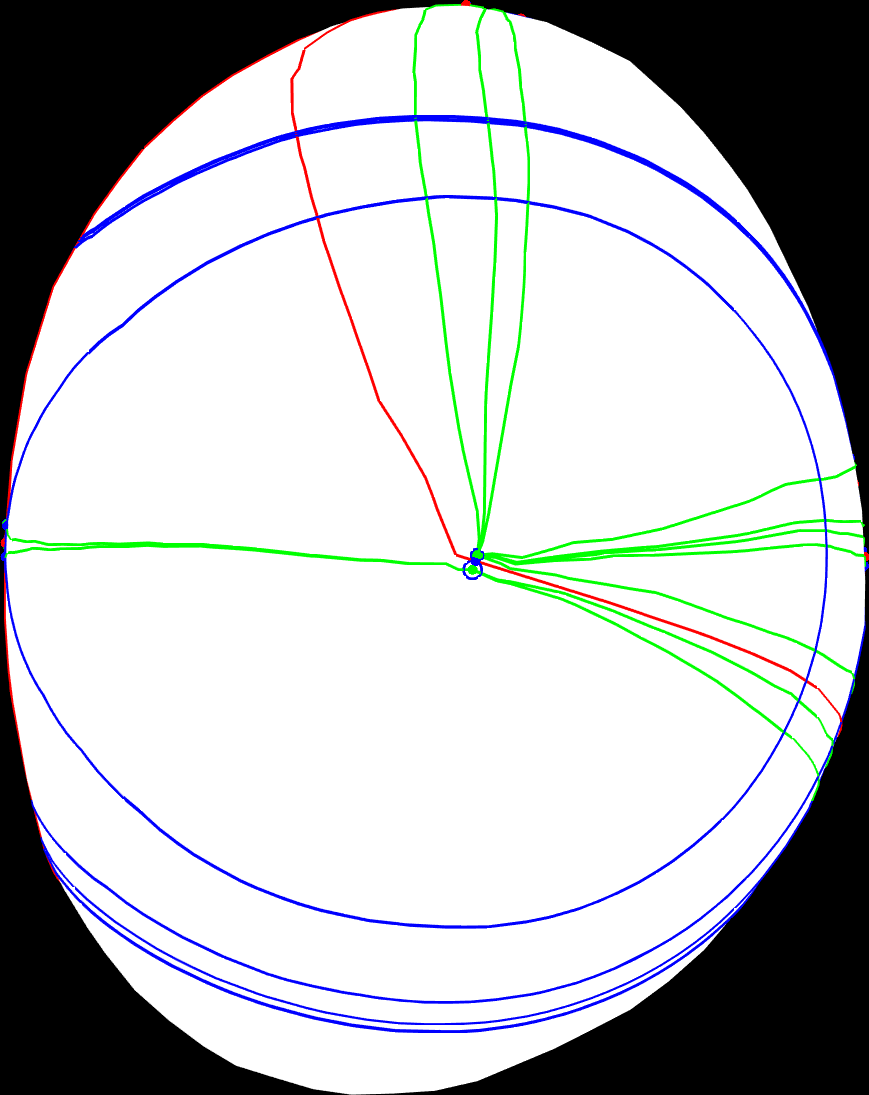}};
        \node at (7,0) {\includegraphics[height=10cm]{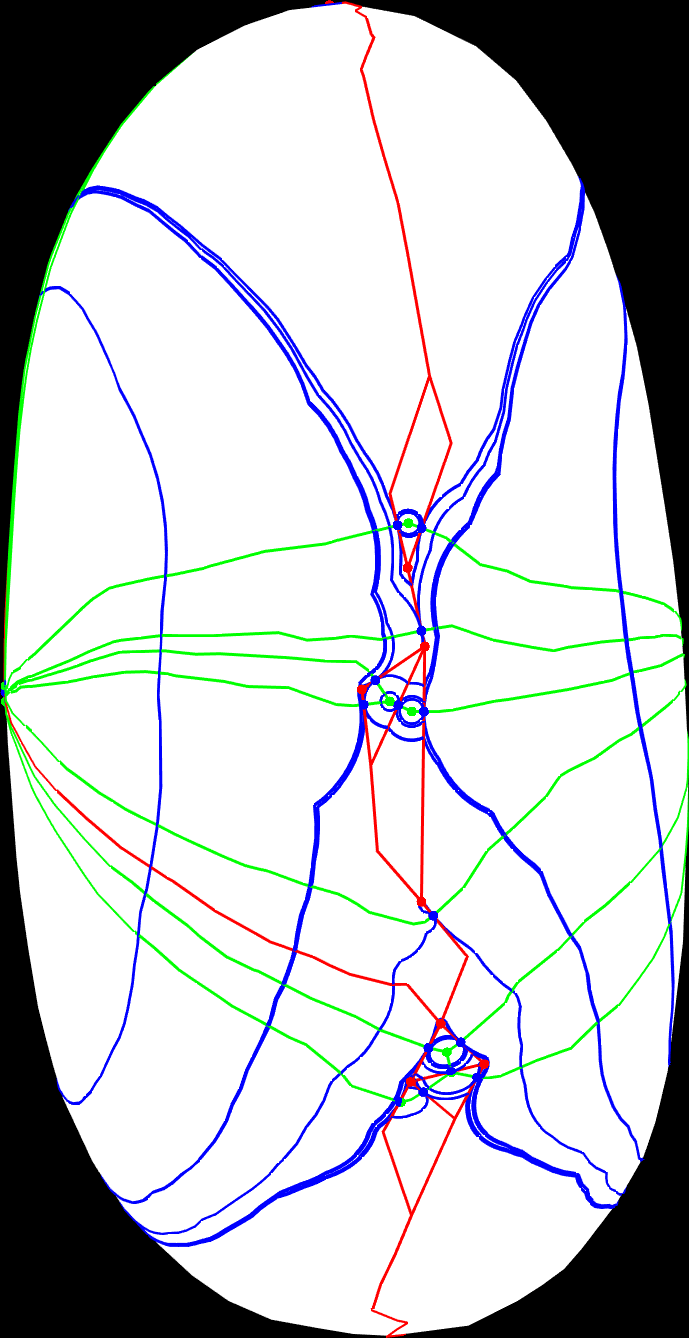}};
    \end{tikzpicture}
    \caption{Saddle contours (blue), saddle\==unstable ascending curves (red) and stable\==saddle ascending curves (green) of the polyhedron $P_\text{ex}$}\label{fig:poly_master}
\end{figure}
\begin{figure}[!ht]
    \centering
    \begin{tikzpicture}
        \node at (0,8) {\includegraphics[width=8cm]{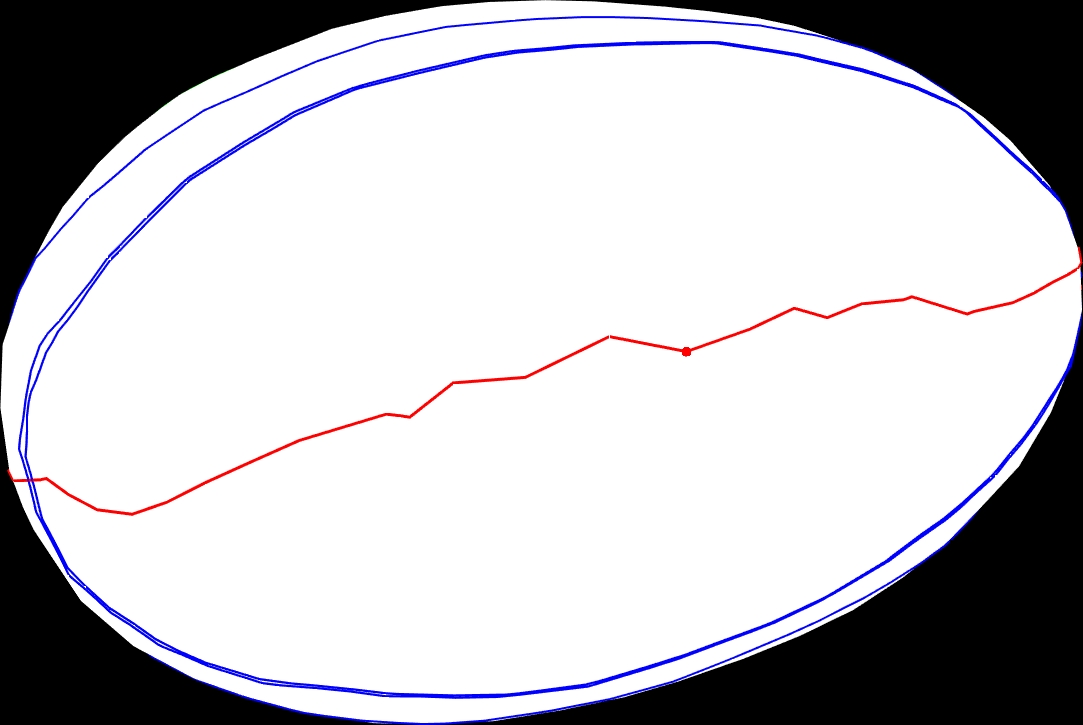}};
        \node at (0,0) {\includegraphics[width=8cm]{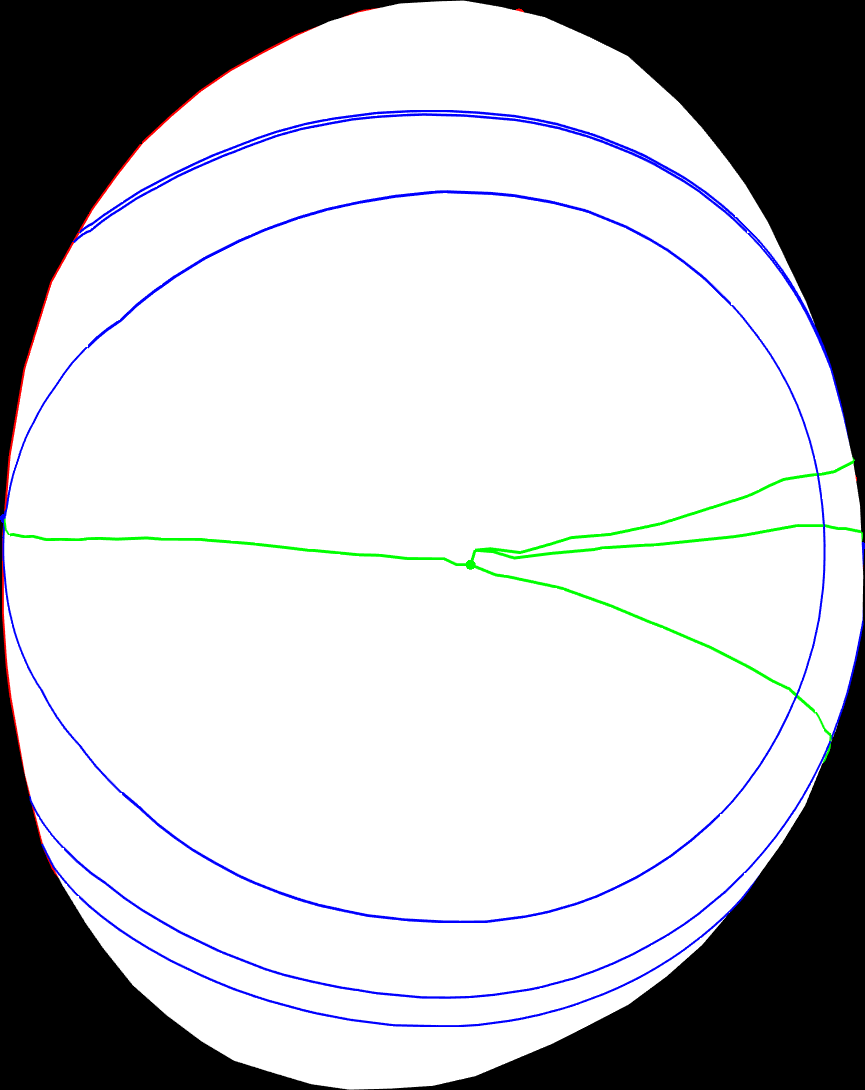}};
        \node at (7,0) {\includegraphics[height=10cm]{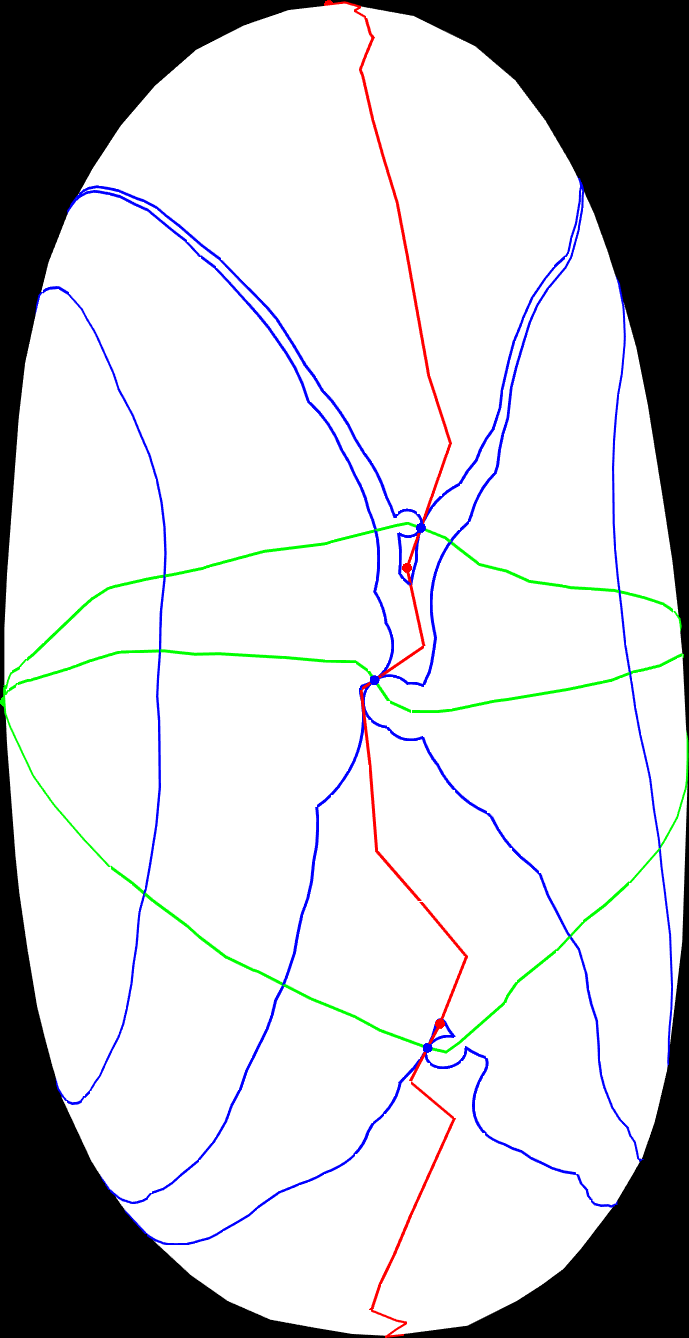}};
    \end{tikzpicture}
    \caption{Saddle contours (blue), saddle\==unstable ascending curves (red) and stable\==saddle ascending curves (green) after cancellations up to the primary class measured by hand}\label{fig:cancelled_master}
\end{figure}

\twocolumn

\addcontentsline{toc}{section}{References}
\printbibliography

\end{document}